\documentclass{article}
\usepackage[english]{babel}
\usepackage{latexsym,amssymb,amsmath,amsfonts,epsfig,color,dsfont,mathtools}

 \usepackage{braket,graphicx,authblk,colonequals,algorithm,algpseudocode,ifthen,framed, comment}
 \usepackage{algpseudocode}
\usepackage{float,array}
\usepackage{xcolor}
\usepackage{tikz}
\usetikzlibrary{quantikz}
\usepackage[utf8]{inputenc}
\usepackage[english]{babel}
\usepackage{amsmath,amsthm}
\usepackage{cite}
\usepackage[colorlinks=true,linkcolor=blue,urlcolor=blue,citecolor=blue]{hyperref}
\usepackage[nameinlink,capitalise]{cleveref}
\floatname{algorithm}{Algorithm}
\usepackage[inline]{enumitem}
\newlist{inlinelist}{enumerate*}{1}
\setlist[inlinelist]{label=(\roman*)}

\newif\ifSUBMISSION

\SUBMISSIONfalse

\ifSUBMISSION

\newcommand{\rgnote}[1]{}

\newcommand{\anote}[1]{}

\newcommand{\Anne}[1]{}

\else

\newcommand{\anote}[1]{\textcolor{purple}{\small (Anne: #1)}}

\newcommand{\Anne}[1]{{\color{red}\sc note}\footnote{\color{red} \underline{{\sc anne}:} #1}}

\fi

\newcommand{\bit}[1]{\{0,1\}^{#1}}

\theoremstyle{definition}
\newtheorem{definition}{Definition}[section]
\newtheorem{remark}[definition]{Remark}
\newtheorem{note}[definition]{Note}
\newtheorem{theorem}[definition]{Theorem}
\newtheorem{claim}[definition]{Claim}
\newtheorem{lemma}[definition]{Lemma}

\definecolor{darkgreen}{rgb}{0,.5,0}

\def\ket#1{{\lvert}#1\rangle}

\DeclareMathOperator{\Tr}{Tr}

\newcommand{\igate}{{\sf I}}
\newcommand{\hgate}{{\sf H}}
\newcommand{\pgate}{{\sf P}}
\newcommand{\tgate}{{\sf T}}
\newcommand{\xgate}{{\sf X}}
\newcommand{\ygate}{{\sf Y}}
\newcommand{\zgate}{{\sf Z}}
\newcommand{\cnot}{{\sf CNOT}}

{\title{Constructions for Quantum \\Indistinguishability Obfuscation
}

\author{Anne Broadbent and Raza Ali Kazmi\footnote{University of Ottawa, Department of Mathematics and Statistics. \texttt{(abroadbe,rkazmi)@uottawa.ca}}}

\date{}
\begin{document}
\maketitle

\begin{abstract}
An \emph{indistinguishability obfuscator} is a probabilistic polynomial-time algorithm that takes a circuit as input and outputs a new circuit  that has the same functionality as the input circuit, such that  for any two circuits of the same size that compute the \emph{same} function, the outputs of the indistinguishability obfuscator are indistinguishable. Here, we study schemes for indistinguishability obfuscation for \emph{quantum} circuits. We present two definitions for indistinguishability obfuscation:
in our first definition  ($qi\mathcal{O}$) the outputs of the obfuscator are required to be indistinguishable if the input circuits are perfectly equivalent, while in our second definition ($qi\mathcal{O}_{\bf D}$), the outputs are required to be indistinguishable as long as the input circuits are approximately equivalent with respect to a pseudo-distance~{\bf D}. Our main results provide (1) a computationally-secure scheme for $qi\mathcal{O}$ where the size of the output of the obfuscator is exponential in the number of non-Clifford ($\tgate$ gates), which means that the construction is efficient as long as the number of
$\tgate$ gates is logarithmic in the circuit size and (2) a statistically-secure $qi\mathcal{O}_{\bf D},$ for circuits that are close to the $k$th level of the Gottesman-Chuang hierarchy (with respect to {\bf D});  this construction is efficient as long as $k$ is small and fixed.

\end{abstract}

\section{Introduction}
At the intuitive level, an \emph{obfuscator} is a probabilistic polynomial-time algorithm that  transforms a circuit $C$ into another circuit $C’$ that has the same functionality as $C$ but that does not reveal anything about $C$, except its functionality \emph{i.e.}, anything that can be learned from $C’$ about $C$ can also be learned from  black-box access to the input-output functionality of $C$. This concept is formalized in terms of \emph{virtual black-box obfuscation}, and was shown~\cite{BGI+12} to be unachievable in general.
Motivated by this impossibility result, the same work proposed a weaker notion called \emph{indistinguishability obfuscation} ($i\mathcal{O}$).

In the classical case, an \emph{indistinguishability obfuscator} is a probabilistic polynomial-time algorithm that takes a circuit~$C$ as  input and outputs a circuit~$i\mathcal{O}(C)$ such that $i\mathcal{O}(C)(x)=C(x)$ for all inputs~$x$ and the size of $i\mathcal{O}(C)$ is at most polynomial in the size of~$C$. Moreover, it must be that for any two circuits $C_1$ and~$C_2$ of the same size and that compute the same function, their obfuscations are computationally indistinguishable. It is known that $i\mathcal{O}$ achieves the notion of \emph{best possible obfuscation}, which states that any information that is not hidden by the obfuscated circuit is also not hidden by any  circuit of similar size computing the same functionality~\cite{GR14}. Indistinguishability obfuscation is a very powerful cryptographic tool which is known to enable, among others: digital signatures, public key encryption~\cite{SW14}, multiparty key agreement, broadcast encryption~\cite{BZ14}, fully homomorphic encryption~\cite{CLTV15} and witness-indistinguishable proofs~\cite{BP15}.
Notable in the context of these applications is the \emph{punctured programming technique}~\cite{SW14} which manages to render an $i\mathcal{O}(C)$ into an intriguing cryptographic building block, and this, despite that fact that
 the security guarantees of $i\mathcal{O}(C)$ appear quite weak as they are applicable only if the two original circuits have \emph{exactly} the same functionality.

The first candidate construction of $i\mathcal{O}$ was published in~\cite{GGH+13}, with security relying on the presumed hardness of multilinear maps~\cite{CLT13,LSS14,GGH15}. Unfortunately, there have been many quantum attacks on multilinear maps~\cite{ABD16,CDPR16,CGH17}.  Recently, new $i\mathcal{O}$ schemes were proposed under  different assumptions ~\cite{AJL+19, JLS20,GJLS21}. Whether or not these schemes are resistant against quantum attacks remains to be determined.

Indistinguishability obfuscation has been studied for \emph{quantum} circuits in~\cite{AJJ14,AF16arxiv}. In a nutshell (see \Cref{sec:obf-quantum} for more details), ~\cite{AJJ14} shows a type of obfuscation for quantum circuits, but without a security reduction.  On the other hand,  the focus of ~\cite{AF16arxiv} is on impossibility of obfuscation for quantum circuits in a variety of scenarios. Thus, despite these works, until now, the achievability of indistinguishability obfuscation for quantum circuits has remained wide open.

\subsection{Overview of Results and Techniques}
\label{sec:techniques}

Our contribution establishes indistinguishability obfuscation for certain families of quantum circuits. We now overview each of our two main definitions, and methods to achieve them (\Cref{sec:intro-summary-results-first-definition} and \Cref{sec:intro-summary-results-second-definition}). We then compare the two approaches (\Cref{sec:intro-comparison}).

\subsubsection{Indistinguishability obfuscation for quantum circuits}
\label{sec:intro-summary-results-first-definition}
First, we define indistinguishability obfuscation for quantum circuits ($qi\mathcal{O}$) (\Cref{sec:definitions}) as an extension of the conventional classical definition. This definition specifies that on input a classical description of a quantum circuit~$C_q$, the obfuscator outputs a \emph{pair} $(\ket{\phi}, C_q^\prime)$, where $\ket{\phi}$ is an auxiliary quantum state and $C_q^\prime$ is a quantum circuit. For correctness, we require that  $ ||C_q^\prime(\ket{\phi}, \cdot)-C_{q}(\cdot) ||_\diamond=0$, whereas for security,
we require that, on input two functionally equivalent quantum circuits, the outputs of $qi\mathcal{O}$ are indistinguishable.
As a straightforward extension of the classical results, we then argue that \emph{inefficient} indistinguishability obfuscation exists.

In terms of constructing $qi\mathcal{O}$, we first focus on the family of \emph{Clifford} circuits and show two methods of obfuscation: one straightforward method based on the canonical representation of Cliffords, and another based on the
principle of gate teleportation~\cite{GC99}. Clifford circuits are quantum circuits that are built from the gate-set $\{\xgate, \zgate, \pgate, \cnot, \hgate\}$. They are known not to be universal for quantum computation and  are, in a certain sense, the quantum equivalent of classical \emph{linear circuits}. It is known that Clifford circuits can be efficiently simulated on a classical computer~\cite{Got98}; however, note that  this simulation is with respect to a \emph{classical} distribution, hence for a purely quantum computation, quantum circuits are required, which motivates the obfuscation of this circuit class. Furthermore,  Clifford circuits are an important building block for fault-tolerant quantum computing, for instance, due to the fact that Cliffords admit transversal computations in many fault-tolerant codes. We provide two methods to achieve $qi\mathcal{O}$ for Clifford circuits.

\paragraph{Obfuscating Cliffords using a canonical form.} Our first construction of~$qi\mathcal{O}$ for Clifford circuits starts with the well-known fact that a canonical form is an $i\mathcal{O}$. We point out that a canonical form for Clifford circuits was presented in~\cite{AG04}; this completes this construction (we also note that an alternative canonical form was also presented in~\cite{Sel13arxiv}). This canonical form technique  does not require any computational assumptions. Moreover, the obfuscated circuits are classical, and hence can be easily communicated, stored, used and copied.

\paragraph{Obfuscating Cliffords using gate teleportation.} Our second construction of $qi\mathcal{O}$ for Clifford circuits takes a very different approach. We start with the gate teleportation scheme~\cite{GC99}: according to this, it is possible to \emph{encode} a quantum computation~$C_q$ into a quantum state (specifically, by preparing a collection of entangled qubit pairs, and applying~$C_q$ to half of this preparation). Then, in order to perform a quantum computation on a target input $\ket{\psi}$, we \emph{teleport}~$\ket{\psi}$ \emph{into} the prepared entangled state. This causes the state $\ket{\psi}$ to undergo the evolution of~$C_q$, \emph{up to some corrections}, based on the teleportation outcome. If~$C_q$ is chosen from the Clifford circuits, these corrections are relatively simple\footnote{The correction is a tensor products of \emph{Pauli} operators, which is computed as a function of $C_q$ and of the teleportation outcome.} and thus we can use a  classical  $i\mathcal{O}$ to provide the correction function.
In contrast to the previous scheme, the gate teleportation scheme requires the assumption of quantum-secure classical $i\mathcal{O}$ for a certain family of functions (\cref{update function}) and the obfuscated circuits include a quantum system. While this presents a technological challenge to communication, storage and also usage, there could be advantages to storing quantum programs into quantum states, for instance to take advantage of their \emph{uncloneability} ~\cite{Aar09,BL19arxiv}.

\paragraph{Obfuscating Beyond Cliffords.} Next, in our main result for  \Cref{QiO:Clifford+T:family}, we generalize the gate teleportation scheme for Clifford circuits, and show a $qi\mathcal{O}$ obfuscator for all quantum circuits where the number of non-Clifford  gates is at most logarithmic in the circuit size. For this, we consider the commonly-used Clifford+$\tgate$ gate-set, and we note that the $\tgate$ relates to the $\xgate, \zgate$ as: $\tgate \xgate^b \zgate^a=\xgate^b \zgate^{a\oplus b}\pgate^b \tgate$. This means that, if we implement a circuit $C$ with $\tgate$ gates as in the gate teleportation scheme above, then the \emph{correction} function is no longer a simple Pauli update (as in the case for Cliffords). However, this is only partially true: since the Paulis form a basis, there is always a way to represent an update as a complex, linear combination of Pauli matrices. In particular, for the case of a~$\tgate$, we note that $\pgate=(\frac{1+i}{2}) \igate + (\frac{1-i}{2})\zgate$. Hence, it \emph{is} possible to produce an update function for general quantum circuits that are encoded via gate teleportation. To illustrate this, we first analyze the case of a general Clifford+$\tgate$ quantum circuit on a \emph{single} qubit (\Cref{sec:1-qubit}). Here, we are able to provide $qi\mathcal{O}$ for all circuits. Next, for general quantum circuits, (\Cref{sec:n-qubit-circuits}), we note that the update function exists for all circuits, but becomes more and more complex as the number of $\tgate$ gates increases. We show that if we limit the number of $\tgate$ gates to be logarithmic in the circuit size, we can reach an efficient construction. Both of these constructions assume a quantum-secure, classical indistinguishability obfuscation.

To the best of our knowledge, our gate teleportation provides the first method for indistinguishability obfuscation that is efficient for a large class of quantum circuits, beyond Clifford circuits. Note, however that canonical forms (also called \emph{normal} forms) are known for \emph{single}-qubits universal quantum circuits~\cite{MA08arxiv,GS19arxiv}.
We note that,  for many other quantum cryptographic primitives, it is the case that the $\tgate$-gate is the bottleneck (somewhat akin to a \emph{multiplication} in the classical case). This has been observed, \emph{e.g.}, in the context of \emph{homomorphic quantum encryption}~\cite{BJ15,DSS16}, and instantaneous quantum computation~\cite{Spe16}.
Because of these applications, and since the $\tgate$ is also typically also the bottleneck for fault-tolerant quantum computing, techniques exist to reduce the number of $\tgate$ gates in quantum circuits~\cite{AMMR13,AMM14,DMM16} (see \Cref{sec:intro-related-work} for more on this topic).

\subsubsection{Indistinguishability obfuscation for quantum circuits, with respect to a pseudo-distance}
\label{sec:intro-summary-results-second-definition}
Next in \Cref{sec:quantum:iO:approx:circuits}, we define indistinguishability obfuscation for quantum circuits with respect to some pseudo-norm {\bf D}, which we call~$qi\mathcal{O}_{\bf D}$. This definition specifies that on input a classical description of a quantum circuit~$C_q$, the obfuscator outputs a \emph{pair} $(\ket{\phi}, C_q^\prime)$, where~$\ket{\phi}$ is an auxiliary quantum state and $C_q^\prime$ is a quantum circuit. For correctness, we require that
${\bf D}(C_q^\prime(\ket{\phi}, \cdot),C_{q}(\cdot))\leq {\tt negl}(n)$, whereas for security, we require that, on input two \emph{approximately} equivalent quantum circuits (\Cref{def:aqec}), the outputs of $qi\mathcal{O}_{\bf D}$ are statistically indistinguishable. This definition is more in line with~\cite{AF16arxiv}.

We show how to construct a statistically-secure quantum indistinguishability obfuscation with respect to the pseudo-distance ${\bf D}$ (see \Cref{QiO:gottesman-chuang}) for quantum circuits that are very close to $k$th level of the  Gottesman-Chuang hierarchy~\cite{GC99}, for some fixed $k$ (see \Cref{sec:gottesman-chuang}). The construction takes a circuit~$U_q$ as an input with a promise that the distance ${\bf D}(U_q, C)\leq \epsilon <\frac{1}{2^{k+1/2}}$  for some $C\in \mathcal{C}_k.$ It computes the conjugate circuit $U_q^\dagger$ and then runs Low's learning algorithm as a subroutine on inputs $U_q$ and $U_q^\dagger$ \cite{Low09}. The algorithm outputs whatever Low's learning algorithm outputs. Note that  Low's learning algorithm runs in time super-polynomial in $k,$ therefore for our construction to remain efficient the parameter $k$ is some small fixed integer (say $k=5$). Note that for $k>2$, the set $\mathcal{C}_k$ includes all Clifford unitaries as well as some non-Clifford unitaries~\cite{Low09}.

\subsubsection{Comparison of the two Approaches}
\label{sec:intro-comparison}
Our notions of  $qi\mathcal{O}$ and  $qi\mathcal{O}_{{\bf D}}$ are incomparable.
 To see this, on one hand, note that the basic  instantiation of an indistinguishability obfuscator that outputs a canonical form is no longer  secure in the  definition of indistinguishability with respect to a pseudo-norm.\footnote{If two different circuits are close in functionality but not identical, then we have no guarantee that their canonical forms are close.} On the other hand, the construction for $qi\mathcal{O}_{{\bf D}}$ that we give in \Cref{QiO:gottesman-chuang} does not satisfy the definition of $qi\mathcal{O}$, because the functionality is not perfectly preserved, which is a requirement for $qi\mathcal{O}$.
  We recall that in the classical case, it is generally considered an \emph{advantage} that  $i\mathcal{O}$ is a relatively weak notion (since it is more easily attained) and that, despite this, a host of uses of $i\mathcal{O}$ are known.  We thus take $qi\mathcal{O}$ as the more natural extension of classical indistinguishability obfuscation to the quantum case, but we note that issues related to the continuity of quantum mechanics and the inherent approximation in any universal quantum gateset justify the relevance for our approach to $qi\mathcal{O}_{{\bf D}}$.

We now compare the schemes that we achieve. The most general scheme that we give as a construct for $qi\mathcal{O}$ (\Cref{QiO:qcircuit-teleportation}) allows  to obfuscate any polynomial-size quantum circuit (with at most $O(\log)$ non-Clifford gates). While this is a restricted class, it is well-understood and we believe that this technique may be amenable to an extension that would result into a full $qi\mathcal{O}$.

In comparison, the scheme that we give for $qi\mathcal{O}_{{\bf D}}$, based on Low's learning algorithm \cite{Low09}
 has some advantages over the teleportation-based constructions. Firstly,  the circuits to be obfuscated don't need to be of equal size or perfectly equivalent and  the outputs of the obfuscator remain statistically indistinguishable as long as the circuits are approximately equivalent (with respect to the pseudo-distance {\bf D}). Secondly, \Cref{QiO:gottesman-chuang} does not require any computational assumptions, whereas the teleportation-based constructions require a quantum-secure classical indistinguishability obfuscator. However, beyond the fact that~$\mathcal{C}_k$ contains all Clifford circuits, it is not clear how powerful unitaries are in the $k$th level of the Gottesman-Chuang hierarchy (especially for a fixed small~$k$). Even when $k\rightarrow \infty$, the hierarchy does not include all unitaries. In terms of extending this technique, Low's learning algorithm exploits the structure of the Gottesman-Chuang hierarchy and it not obvious how one can apply this technique to arbitrary quantum circuits.

\subsection{More on Related Work}
\label{sec:intro-related-work}

\paragraph{Quantum Obfuscation.}
\label{sec:obf-quantum}
Quantum obfuscation was first studied in ~\cite{AJJ14}, where a notion called
$(G,\Gamma)$-{\em indistinguishability obfuscation} was proposed,  where $G$ is a set of gates and $\Gamma$ is a set of relations satisfied by the elements of~$G.$ In this notion, any two circuits over the set of gates $G$ are perfectly indistinguishable if they differ by some sequence of applications of the relations in~$\Gamma.$

Since perfect indistinguishability obfuscation is known to be impossible under the assumption that $\mathsf{P} \neq \mathsf{NP}$ \cite{GR14},
one of the motivations of this work was to provide a weaker definition of perfectly indistinguishable obfuscation, along with possibility results. However, to the best of our knowledge, $(G,\Gamma)$-{\em indistinguishability obfuscation} is incomparable with computational indistinguishability obfuscation \cite{BGI+12,GGH+13}, which is the main focus  of our work.

Quantum obfuscation is studied in ~\cite{AF16arxiv}, where the various notions of quantum obfuscation are defined (including quantum black-box obfuscation, quantum indistinguishability obfuscation, and quantum best-possible obfuscation).
A contribution of ~\cite{AF16arxiv} is to extend the classical impossibility results to the quantum setting, including \emph{e.g.} showing that
 each of the three variants of quantum indistinguishability obfuscation is equivalent to the analogous variant of quantum best-possible
obfuscation, so long as the obfuscator is efficient.  This work shows that the existence of a computational quantum indistinguishability obfuscation implies a witness encryption scheme for all languages in~\textsf{QMA}. Various impossibiliity results are also shown: that efficient statistical indistinguishability obfuscation is impossible unless \textsf{PSPACE} is contained in \textsf{QSZK}\footnote{\textsf{PSPACE} is the class of decision problems solvable by a Turing machine in polynomial space and  \textsf{QSZK} is the class of decision problems that admit a quantum statistical zero-knowledge proof system.} (for the case of circuits that include measurements), or unless \textsf{coQMA}\footnote{\textsf{coQMA} is the \emph{complement} of \textsf{QMA}, which is the class of decision problems that can be verified by a one-message quantum interactive proof.} is contained in \textsf{QSZK} (for the case of unitary circuits). Notable here is that \cite{AF16arxiv}
  defines a notion of indistinguishability  obfuscation where security must hold for circuits that are \emph{close} in functionality (this is similar to our definition of $qi\mathcal{O}_{\bf D}$);  it is however unclear if their  impossibility results  hold for a notion of quantum indistinguishability along the lines of our definition of $qi\mathcal{O}$. See \Cref{sec:defs-qiO} for further discussion of the links between this definition and ours. We note that \cite{AF16arxiv} does not provide any concrete instantiation of obfuscation.

  Recently it has been shown that virtual black-box obfuscation of classical circuits via quantum mechanical means is also impossible ~\cite{ABDS20,AP20}.

\paragraph{Quantum Homomorphic Encryption.}
In \emph{quantum homomorphic encryption}, a computationally-weak client is able to send a ciphertext to a quantum server, such that the quantum server can perform a quantum computation on the encrypted data, thus producing an encrypted output which the client can decrypt, and obtaining the result of the quantum computation.

This primitive was formally defined in~\cite{BJ15} (see also ~\cite{DSS16,Bra18}), where it was shown how to achieve homomorphic quantum computation for quantum circuits of  low $\tgate$-depth, by assuming quantum-secure classical fully homomorphic encryption.  We note that even the simplest scheme in~\cite{BJ15} (which allows the homomorphic evaluation of \emph{any} Clifford circuit), requires computational assumptions in order for the server to update homomorphically the classical portion of the ciphertext, based on the choice of Clifford. In contrast, here we are able to give information-theoretic constructions for this class of circuits (essentially, because the choice of Clifford is chosen by the obfuscator, not by the evaluator). We thus emphasize that in $i\mathcal{O}$, we want to hide the \emph{circuit}, whereas in homomorphic encryption, we want to hide the \emph{plaintext} (and allow remote computations on the ciphertext). Since the evaluator in homomorphic encryption has control of the circuit, but not of the data, the evaluator knows which types of gates are applied, and the main obstacle  is to perform a correction after a $\tgate$-gate, controlled on a classical value that is held only in an encrypted form by the evaluator.
In contrast to this, in $i\mathcal{O}$, we want to hide the inner workings of the circuit. By using gate teleportation, we end up in a situation where the evaluator \emph{knows} some classical values that have affected the quantum computation in some undesirable way, and then we want to hide the inner workings of \emph{how} the evaluator should compensate for these undesirable effects.
Thus, the techniques of quantum homomorphic encryption do not seem directly applicable, although we leave as an open question if they could be used in some indirect way, perhaps towards efficient $qi\mathcal{O}$ for a larger family of circuits.

\subsection{Open Questions}
The main open question is efficient quantum indistinguishability obfuscation for quantum circuits with super-logarithmic number of $\tgate$-gates.
Another open  question is about the applications of quantum indistinguishability obfuscation. While we expect that many of the uses of classical $i\mathcal{O}$ carry over to the quantum case, we leave as future work the formal study of these techniques.

\paragraph{Outline.}
The remainder of this paper is structured as follows. \Cref{sec:prelims} overviews basic notions required in this work. In \Cref{sec:definitions}, we formally define indistinguishability obfuscation for quantum circuits. In \Cref{QiO:Clifford-Circuits}, we provide the construction for Clifford circuits. In \Cref{QiO:Clifford+T:family}, we give our main result which shows quantum indistinguishability obfuscation for quantum circuits, which is efficient for circuits having at most a logarithmic number of $\tgate$ gates. Finally in \Cref{sec:quantum:iO:approx:circuits}, we consider the notion of quantum indistinguishability obfuscation with respect to a pseudo-distance, and show how to instantiate it for a family of circuits close to the Gottesman-Chuang hierarchy.  

\section{Preliminaries}
\label{sec:prelims}
\subsection{Basic Classical Cryptographic Notions}
\label{sec:classical-prelims}
Let $\mathbb{N}$ be the set of positive integers. For $n \in \mathbb{N}$, we set $[n] = \{1, \cdots, n\}.$ We denote the set of all  binary strings of length $n$ by $\bit{n}.$
 An element $s \in \bit{n}$ is called a bitstring, and $|s|=n$ denotes its length. Given two bit strings $x$ and~$y$ of equal length, we denote their bitwise XOR by $x \oplus y$. For a finite set $X$, the notation $x \xleftarrow{\text{\$}} X$ indicates that $x$ is selected uniformly at random from~$X$. We denote the set of all $d \times d$ unitary matrices by $\mathcal{U}(d)=\{U \in \mathbb{C}^{d\times d} \mid UU^\dagger={\bf I}\}$, where $U^\dagger$ denotes the conjugate transpose of $U.$

A function $ {\tt negl}:\mathbb{N}\rightarrow\mathbb{R}^{+}\cup \{0\}$ is \emph{negligible} if for every positive polynomial~$p(n)$, there exists a positive integer $n_0$ such that  for all  $n>n_0,$ $ {\tt negl}(n) < 1/ p(n).$   A typical use of negligible functions is to indicate that the probability of success of some algorithm is too small to be amplified to a constant by a feasible (\emph{i.e.}, polynomial) number of repetitions.

\subsection{Classical Circuits and Algorithms}
\label{sec:cir:alg}
A  deterministic polynomial-time (or {\bf PT}) algorithm $\mathcal{C}$ is defined by a polynomial-time uniform\footnote{Recall that polynomial-time uniformity means that there exists a polynomial-time Turing machine which, on input~$n$ in unary, prints a description of the $n$th circuit in the family.} family $\mathcal{C}=\{C_{n}\mid n\in \mathbb{N} \}$ of classical Boolean circuits over some gate set, with one circuit for each possible input size $n\in\mathbb{N}.$ For a bitstring $x$, we define $\mathcal{C}(x) := \mathcal{C}_{|x|}(x)$. We say that a function family $f:\{0,1\}^n \rightarrow \{0,1\}^m$ is {\bf PT}-computable if there exists a polynomial-time $\mathcal{C}$ such that $\mathcal{C}(x) = f(x)$ for all~$x$; it is implicit that $m$ is a function of $n$ which is bounded by some polynomial, \emph{e.g.}, the same one that bounds the running time of $\mathcal{C}.$ Note that in the literature, circuits that compute functions whose range is $\{0,1\}^m$ are often called multi-output Boolean circuits \cite{GMOR15}, but in this paper we simply called them Boolean circuits \cite{Sip12}.

A probabilistic polynomial-time algorithm (or {\bf PPT}) is again a polynomial-time uniform family of classical Boolean circuits, one for each possible input size~$n.$ The $n$th circuit still accepts $n$ bits of input, but now also has an additional ``coins'' register of $p(n)$ input wires. Note that uniformity enforces that the function $p$ is bounded by some polynomial. For a {\bf PPT} algorithm~$\mathcal{C},$ $n$-bit input $x$ and $p(n)$-bit coin string $r$, we set $\mathcal{C}(x; r) := \mathcal{C}_n(x; r).$ In contrast with the PT case, the notation algorithm $\mathcal{C}(x)$ will now refer to the random variable algorithm $\mathcal{C}(x; r)$ where $r \xleftarrow{\text{\$}} \{0,1\}^{p(n)}.$

\subsection{Classical Indistinguishability}

Here, we define indistinguishability for classical random variables, against a quantum distinguisher (\Cref{def:classical-indis}).

\begin{definition} (Statistical Distance)
Let  $X$ and $Y$ be two random variables over some countable set $\Omega$. The statistical distance between $X$ and~$Y$ is

\begin{center}
	$\Delta(X,Y)=\frac{1}{2}\left\{\sum_{\omega\in \Omega} \left |Pr[X(\omega)]- Pr[Y(\omega)]\right | \right\}.$
\end{center}
\end{definition}

\begin{definition} (Indistinguishability) \label{def:classical-indis}
Let $\mathcal{X}=\{X_n\}_{n\in\mathbb{N}}$ and $\mathcal{Y}=\{Y_n\}_{n\in\mathbb{N}}$ be two distribution ensembles indexed by a parameter $n.$ We say
\begin{enumerate}
\item  $\mathcal{X}$ and $\mathcal{Y}$ are \emph{perfectly indistinguishable} if for all $n,$  $$\Delta(X_n,Y_n)=0.$$
\item $\mathcal{X}$ and $\mathcal{Y}$ are \emph{statistically indistinguishable} if there exists a negligible function {\tt negl} such that  for all sufficiently large $n$: $$\Delta(X_n,Y_n)\leq  {\tt negl}(n).$$
\item $\{X_n\}_{n\in\mathbb{N}}$ and $\{Y_n\}_{n\in\mathbb{N}}$ are \emph{computationally indistinguishable} if for any polynomial-time quantum distinguisher $\mathcal{D}_q$, there exists a negligible function {\tt negl} such that:
$$\Big |{\rm Pr}[\mathcal{D}_q(X_n)=1]-{\rm Pr}[\mathcal{D}_q(Y_n)=1] \Big |\leq   {\tt negl}(n).$$
\end{enumerate}
\end{definition}

\subsection{Classical Indistinguishability Obfuscation}
\label{def:iO}

Let $\mathcal{C}$ be a family of probabilistic polynomial-time circuits. For $n\in\mathbb{N},$ let~$C_n$ be the circuits in $\mathcal{C}$ of input length~$n.$ We now provide a definition of classical  indistinguishability obfuscation ($i\mathcal{O}$) as defined in \cite{GR14}, but where we make a few minor modifications.\footnote{We make a few design choices that are more appropriate for our situation, where we show the \emph{possibility} of  $i\mathcal{O}$ against quantum adversaries:  our adversary is a probabilistic polynomial-time quantum algorithm, we dispense with the mention of the random oracle, and note that our indistinguishability notions are defined to hold for all inputs.}

\begin{definition}\label{def:quantum-secureiO} {\rm({\bf Indistinguishability Obfuscation}, $i\mathcal{O}$)}
A probabilistic polynomial-time algorithm is a \emph{quantum-secure indistinguishability obfuscator} ($i\mathcal{O}$) for a class of circuits ${\mathcal C},$ if the following conditions hold:
\begin{enumerate}
\item  {\tt Preserving Functionality:} For any $C\in C_n:$
													$$i\mathcal{O}(x)=C(x), \mbox{ for all } x \in \{0,1\}^n$$	
The probability is taken over the $i\mathcal{O}$'s coins.

\item  {\tt Polynomial Slowdown:} There exists a polynomial $p(n)$ such that for all input lengths, for any $C\in C_n,$ the obfuscator $i\mathcal{O}$ only enlarges $C$ by a factor of
$p(|C|):$
													$$ |i\mathcal{O}(C)| \leq p(|C|).$$

\item {\tt Indistinguishability:}  An $i\mathcal{O}$ is said to be a computational/statistical/\\perfect indistinguishability obfuscation for the family $\mathcal{C},$ if for all large enough input lengths, for any circuit $C_1\in C_n$ and for any $C_2\in C_n$ that computes the same function as $C_1$ and such that $|C_1|=|C_2|,$ the distributions $i\mathcal{O}(C_1))$ and $i\mathcal{O}(C_2)$ are (respectively) computationally/statistically/perfectly indistinguishable.
\end{enumerate}										
\end{definition}

\subsection{Basic Quantum Notions}
\label{sec:quantum-prelims}
Given an $n$-bit string $x$, the corresponding  $n$-qubit quantum computational basis state is denoted~$\ket{x}$. The $2^n$-dimensional Hilbert space spanned by $n$-qubit basis states is denoted:
\begin{equation}
\label{eq:hilbert-space}
\mathcal{H}_n := \textbf{span} \left\{ \ket{x} : x \in \bit{n} \right\}\,.
\end{equation}
We denote by $\mathcal{D}(\mathcal{H}_n)$ the set of density operators (\emph{i.e.}, valid quantum states) on~$\mathcal{H}_n$. These are linear operators on $\mathcal{D}(\mathcal{H}_n)$ which are positive-semidefinite and have trace equal to $1$.

\subsection{Norms and Pseudo-Distance}
\label{sec:norms}
The trace distance between two quantum states $\rho, \sigma\in \mathcal{D}(\mathcal{H}_n)$ is given by:
 $$||\rho-\sigma||_{tr}:=\frac{1}{2}\Tr\left(\left\lvert\sqrt{(\rho-\sigma)^\dagger (\rho-\sigma)}\right\rvert\right),$$
 where $\lvert \cdot\rvert$ denotes the positive square root of the matrix $\sqrt{(\rho-\sigma)^\dagger (\rho-\sigma)}.$

Let $\Phi$ and $\Psi$ be two admissible operators of type $(n,m)$\footnote{An operator is admissible if its action on density matrices is linear, trace-preserving, and completely positive. A operator's type is $(n,m)$ if it maps $n$-qubit states to $m$-qubit states.}. The \emph{diamond norm} between two quantum operators is
$$||\Phi-\Psi||_\diamond :=\underset{\rho\in \mathcal{D}(\mathcal{H}_{2n}) }{max} ||(\Phi \otimes I_n) \rho-(\Psi \otimes I_n) \rho||_{tr}$$

The Frobenius norm of a matrix $A\in\mathbb{C}^{n \times m}$ is defined as  $||A||_F=\sqrt{\Tr (AA^\dagger)}.$ Let $U_1, U_2 \in \mathcal{U}(d)$ be two $d \times d$ unitary matrices. The phase invariant distance between $U_1$ and $U_2$ is

$${\bf D}(U_1,U_2)=\frac{1}{\sqrt{2d^2}} ||U_1\otimes U_1^*- U_2\otimes U_2^*||_F$$
		  $$=\sqrt{1-\left|\frac{ {\rm Tr}(U_1U_2^\dagger)}{d}\right|^2}\,,$$
where $U_i^*$ denotes the matrix with only complex conjugated entries and no transposition and $| z |$ denotes the norm of the complex number $z$. Note that~${\bf D}$ is a pseudo-distance since ${\bf D}(U_1,U_2) =0$ does not imply $U_1=U_2,$ but that $U_1$ and $U_2$ are equivalent up to a phase so the difference is unobservable. It is easy to see that {\bf D}  satisfies the axioms of symmetry (${\bf D}(U_1,U_2)={\bf D}(U_2,U_1)$), the triangle inequality (${\bf D}(U_1,U_2)\leq{\bf D}(U_1,U)+{\bf D}(U,U_2)$) and non-negativity (${\bf D}(U_1,U_2)\geq 0$).

\subsection{Bell Basis and Measurement}
\label{sec:bell:basis}
The four states $\{\ket{\beta_{00}}, \ket{\beta_{01}}, \ket{\beta_{10}}, \ket{\beta_{11}}\}$ are called \emph{Bell States} or \emph{EPR pairs} and form an orthonormal basis of $\mathcal{H}_2.$

$$\ket{\beta_{00}}=\frac{1}{\sqrt2}\left(\ket{00}+\ket{11}\right) \hspace{1cm}  {\beta_{01}}=\frac{1}{\sqrt2}\left(\ket{01}+\ket{10}\right)$$
$$\ket{\beta_{10}}=\frac{1}{\sqrt2}\left(\ket{00}-\ket{11}\right) \hspace{1cm}   \ket{\beta_{11}}=\frac{1}{\sqrt2}\left(\ket{01}-\ket{10}\right)$$
We define a generalized Bell state as a tensor product of $n$ Bell states
									$$\ket{\beta_{s}}=\ket{\beta_{a_i,b_i}}^{\otimes_{i=1}^{n}},$$
									
where $s=a_1b_1,\ldots,a_nb_n\in\{0,1\}^{2n}.$ The set of generalized Bell States $\{\ket{\beta_{s}} \mid s\in \{0,1\}^{2n}\}$ forms an orthonormal basis of $\mathcal{H}_n.$
Given a quantum state 				 $$\ket{\psi}=\sum_{s\in\{0,1\}^{2n}} \alpha_{s}\ket{\beta_{s}},$$
 a Bell measurement in the (generalized) Bell basis on the state $\ket{\psi}$ outputs the string $s$ with probability $|\alpha_{s}|^2$ and leaves the system in the state $\ket{\beta_{s}}.$

\subsection{Quantum Gates}
We will work with the following set of unitary gates
$$ \igate = \left[\begin{array}{cc} 1 & 0\\ 0 & 1\end{array}\right],
\quad \xgate = \left[\begin{array}{cc} 0 & 1\\ 1 & 0\end{array}\right],
\quad\ygate = \left[\begin{array}{cc} 0 & i\\ -i & 0\end{array}\right],
\quad\zgate = \left[\begin{array}{cc} 1 & 0\\ 0 & -1\end{array}\right], $$
$$ \quad\hgate = \frac{1}{\sqrt{2}}\left[\begin{array}{cc}1 & 1\\1 & -1\end{array}\right],
\quad\cnot = \left[\begin{array}{cccc} 1 & 0 & 0 & 0\\ 0 & 1 & 0 & 0\\ 0 & 0 & 0 & 1\\ 0 & 0 & 1 & 0\end{array}\right], \text{ and}
\quad\tgate = \left[\begin{array}{cc} 1 & 0\\ 0 & e^{i\pi/4}\end{array}\right].$$

For any single-qubit density operator $\rho \in \mathcal{D} (\mathcal{H}_1)$, we can encrypt it via the \emph{quantum one-time pad} by sampling  uniform bits $s$ and $t$, and producing
 $\xgate^s \zgate^t \rho \zgate^t \xgate^s$. To an observer that has no knowledge of $s$ and $t$, this system is information-theoretically indistinguishable from the state $\mathds{1}_1/2$ (where $\mathds{1}_1$ is the 2 by 2 identity matrix)\cite{AMTW00}.

\subsection{Gottesman-Chuang Hierarchy}
\label{sec:gottesman-chuang}
The $n$-qubit Pauli group $\mathcal{P}_n$ is a multiplicative group of order $4^{n+1}$, defined as:
$$\mathcal{P}_n=\{\alpha_1 P_1 \otimes \cdots \otimes \alpha_n P_n \mid \alpha_i \in \{\pm 1, \pm i\}, P_i\in\{\igate, \xgate, \zgate, \ygate \}\}\,.$$

Let $C_1$ be the Pauli group $\mathcal{P}_n.$  Then the level $\mathcal{C}_k$ of the \emph{Gottesman-Chuang} hierarchy is defined recursively~\cite{GC99}:
$$\mathcal{C}_k=\{U\in \mathcal{U}(2^n): U\mathcal{P}_nU^\dagger \subseteq \mathcal{C}_{k-1}\}.$$
Note that $\mathcal{C}_2$ is the Clifford group and for $k>2$,  $\mathcal{C}_k$ is no longer a group but contains unitaries that contains a universal gate set.

The set of gates $\{\xgate, \zgate, \pgate, \cnot,\hgate\}$ applied to arbitrary wires redundantly generates the \emph{Clifford group}.
We note the following relations between these gates (these relations hold up to \emph{global phase}; in this work, we use the convention that equal signs for pure states and unitaries hold up to global phase.)
$$\xgate\zgate = - \zgate\xgate,\quad \tgate^2=\pgate,\quad\pgate^2=\zgate,\quad\hgate\xgate\hgate=\zgate,\quad \tgate\pgate=\pgate\tgate,\quad\pgate\zgate=\zgate\pgate.$$
Also, for any $a,b\in\{0,1\}$ we have
$\hgate\xgate^b\zgate^a=\xgate^a\zgate^b \hgate$\,.

\subsection{Quantum Circuits and Algorithms}
\label{sec:quantum-algo}
A quantum circuit is an acyclic network of quantum gates connected by wires. The quantum gates represent quantum operations and wires represent the qubits on which gates act. In general, a quantum circuit can have $n$-input qubits and $m$-output qubits for any integer $n, m\geq 0.$ The \emph{$\tgate$-count} is the total number of $\tgate$-gates in a quantum circuit.

A quantum circuit that computes a unitary matrix is called a \emph{reversible  quantum circuit}, \emph{i.e.}, it always possible to uniquely recover the input, given the output. A set of gates is said to be \emph{universal} if for any integer $n \geq 1,$ any $n$-qubit unitary operator can be approximated to arbitrary accuracy by a quantum circuit using only gates from that set \cite{KLM07}. It is a well-known fact that Clifford gates are not universal, but adding any non-Clifford gate, such as $\tgate$, gives a universal set of gates \cite{KLM07}\footnote{In this work, we assume circuits are given in the Clifford + $\tgate$ gateset}.  \emph{Generalized quantum circuits} (which implement \emph{superoperators}) are composed of the unitary gates, together with  trace-out and measurement operations. It is well-known that a generalized quantum circuit can be implemented by  adding auxiliary states to the original system, applying a unitary operation on the joint system, and then tracing out some subsystem~\cite{KLM07}.

A family of generalized quantum circuits $\mathcal{C}=\{C_{q_n}\mid n\in \mathbb{N}\}$,  one for each input  size $n\in \mathbb{N}$, is called \emph{polynomial-time uniform} if there exists a deterministic Turing machine~$M$ such that:
 \begin{inlinelist}
 \item for each $n\in\mathbb{N},$ $M$ outputs a description of $C_{q_n}\in \mathcal{C}$ on input $1^n$; and
 \item  for each $n\in\mathbb{N},$ $M$ runs in $poly(n).$
 \end{inlinelist}
We define a \emph{quantum polynomial-time algorithm} (or QPT) to be a polynomial-time uniform family of generalized quantum circuits.

\subsection{Quantum Indistinguishability}
\label{sec:comp-stat}
Here, we define indistinguishability for indistinguishability for quantum states (\Cref{def:indis-quantum}).

\begin{definition} (Indistinguishability of Quantum States)
\label{def:indis-quantum}
  Let $\mathcal{R}=\{\rho_n\}_{n\in \mathbb{N}}$ and  $\mathcal{S}=\{\sigma_n\}_{n\in \mathbb{N}}$ be two ensembles of quantum states such that $\rho_n$ and $\sigma_n$ are $n$-qubit states. We say
 \begin{enumerate}
 \item $\mathcal{R}$ and $\mathcal{S}$ are \emph{perfectly indistinguishable} if for all $n,$
																$$\rho_n=\sigma_n.$$
 \item $\mathcal{R}$ and $\mathcal{S}$ are \emph{statistically indistinguishable} if there exists a negligible function {\tt negl} such that for all sufficiently large $n$:
 																		$$||\rho_n-\sigma_n||_{tr}\leq  {\tt negl}(n).$$

\item  $\mathcal{R}$ and $\mathcal{S}$ are \emph{computationally indistinguishable} if there exists a negligible function {\tt negl} such that for every state $\rho_n\in \mathcal{R}$, $\sigma_n\in \mathcal{S}$ and for all polynomial-time quantum distinguisher $\mathcal{D}_q$,  we have:
            						 $$\Big |{\rm Pr}[\mathcal{D}_q(\rho_n)=1]-{\rm Pr}[\mathcal{D}_q(\sigma_n)=1] \Big |\leq   {\tt negl}(n).$$
\end{enumerate}
\end{definition}

\subsection{Quantum Teleportation}
Here we  provide a high-level description of quantum teleportation; for a more rigorous treatment see \cite{BBC+93}.
Suppose Alice has a quantum state $\ket{\psi}=\alpha\ket{0}+\beta\ket{1}$\footnote{For simplicity we assume that $\ket{\psi}$ is single-qubit pure state.} that she wants to send  to Bob who is located far away from Alice. One way for Alice to send her qubit to Bob is via the quantum teleportation protocol. For teleportation to work, Alice prepares a $2$-qubit Bell state  $$\ket{\beta_{00}}_{AB}=\frac{1}{\sqrt2}\left(\ket{00}_{AB}+\ket{11}\right)_{AB},$$
and sends physically one of the qubit to Bob and keeps the other to herself (this is what subscript $AB$ means). We can now write the 3-qubit system as
\begin{equation}
\label{eq:sys}
													\ket{\psi}  \otimes \ket{\beta_{00}}_{AB}
\end{equation}																				
Alice now performs a joint measurement on $\ket{\psi}$ and her part of the EPR pair in the Bell basis and obtains the output of the measurement (classical bits $a,b$). After this step, Bob's part of EPR pair has been transformed into the state
$$\xgate^b \zgate^a \ket{\psi}.$$
Alice sends the two classical bits $(a,b)$ to Bob, who performs the correction unitary $Z^aX^b$ to the state he possesses and obtains the state~$\ket{\psi}.$

\subsection{Gate Teleportation}
One of the main applications of quantum teleportation is in fault-tolerant quantum computation~\cite{GC99}. To construct unitary quantum circuits, we need to have  access to some universal set of quantum gates\footnote{$\{\hgate, \tgate, \cnot\}$ is a universal set of quantum gates \cite{KLM07}.}. Suppose we want to evaluate a single-qubit gate on some quantum state $\ket{\psi}.$ If we directly apply $U$ on  $\ket{\psi}$ and~$U$ fails, then it may also destroy the state. Quantum teleportation gives a way of solving this problem. Instead of applying  $U$ directly to $\ket{\psi}$, we can apply $U$ to the system $B$  in \Cref{eq:sys} and then follow the gate teleportation protocol and obtain  $U(\ket{\psi}).$ If $U$ fails, then the Bell state might be destroyed, but there is no harm done, since we can create another EPR pair and try again. The gate teleportation can easily be generalized to evaluate any $n$-qubit Clifford circuit (\Cref{algo:gate-teleport}).

\begin{remark}
 In this section, we  only discuss how to evaluate Clifford gates using  gate teleportation. Note that we can evaluate any unitary circuit using gate teleportation but the correction unitary becomes more complicated (it is no longer a tensor product of Paulis). This is discussed in \Cref{QiO:Clifford+T:family}.
\end{remark}

\begin{algorithm}[]
{\bf Input}: A $n$-qubit Clifford Circuit $C_q$ and $n$-qubit quantum state $\ket{\psi}$
\caption{Gate Teleportation Protocol.}
\label{algo:gate-teleport}
\begin{enumerate}
  \item Prepare a tensor product of $n$  Bell states: $\ket{\beta^{ 2n}}=\ket{\beta_{00}}\otimes \cdots \otimes \ket{\beta_{00}}.$
\item Write the joint system as $\ket{\psi}_C \ket{\beta^{ 2n}}_{AB}.$
\item Apply the circuit $C_q$ on the subsystem $B.$
\item Perform a measurement in the  generalized Bell Basis (generalized Bell measurement) on the system $CA$ and obtain a binary string $a_1b_1,\ldots, a_nb_n.$ The remaining system after the measurement is
\begin{equation}
 \label{prl:eq2:gate-teleport}
  C_q \left({\xgate^{b_i} \zgate^{a_i}}\right)^{\otimes_{i=1}^{n}}\ket{\psi}.
\end{equation}
\item Compute the correction bits using the update function $F_{C_q}$ (\Cref{update function}).
\begin{equation}\label{algo:update:func}
              F_{C_q}(a_1b_1,\dots,a_nb_n)= a_1^\prime b_1^\prime,\dots,a_n^\prime b_n^\prime \in\{0,1\}^{2n}.
\end{equation}
\item Compute the correction unitary $U_{F_{C_q}}=\left({\zgate^{a_i^\prime} \xgate^{b_i^\prime}}\right)^{\otimes_{i=1}^{n}}$
\item Apply $U_{F_{C_q}}$ to the system (\Cref{prl:eq2:gate-teleport}).
 \begin{equation}
  \label{prl:eq4:gate-teleport}
  \begin{aligned}
  &U_{F_{C_q}} \cdot  C_q \left({\xgate^{b_i} \zgate^{a_i}}\right)^{\otimes_{i=1}^{n}}\ket{\psi}=\left({\zgate^{a_i^\prime} \xgate^{b_i^\prime}}\right)^{\otimes_{i=1}^{n}} C_q \left({\xgate^{b_i} \zgate^{a_i}}\right)^{\otimes_{i=1}^{n}}\ket{\psi}\\
  &=\left({\zgate^{a_i^\prime} \xgate^{b_i^\prime}}\right)^{\otimes_{i=1}^{n}}  \left({\xgate^{b_i^\prime} \zgate^{a_i^\prime}}\right)^{\otimes_{i=1}^{n}}C_q(\ket{\psi})\\
 & =C_q( \ket{\psi}).
  \end{aligned}
 \end{equation}
\end{enumerate}	
\end{algorithm}

\subsection{Update Functions for Quantum Gates}\label{update function}
Let $C_q$ be an $n$-qubit circuit consisting of a sequence of Clifford gates $g_1,\ldots, g_{|C_q|}.$ Then the update function for $C_q$ is a map from $\{0,1\}^{2n}$ to $\{0,1\}^{2n}$ and is constructed by composing the update functions for each gate in $C_q$\footnote{This composition implicitly assumes that when an update function is applied, it acts non-trivially on the appropropriate bits, as indicated by the original circuit, and as the identity elsewhere.}

\begin{equation}
\begin{aligned}
&F_{C_q}=\{0,1\}^{2n} \longrightarrow \{0,1\}^{2n}\\
& F_{C_q}= f_{g_{{|C_q|}}}\circ  \cdots \circ f_{g_2} \circ f_{g_1}
\end{aligned}
\end{equation}

For each Clifford gate $g,$ the update function $f_g$ is defined below. Note how~$g$ relates to the $\xgate$ and $\zgate$ gates.
\begin{equation*}
\begin{aligned}
&\xgate (\xgate^b\zgate^a)\psi=(\xgate^b\zgate^a) \xgate \ket{\psi}  \mbox{ (update function) }  f_\xgate(a,b)=(a,b)\\
&\zgate (\xgate^b\zgate^a)\zgate=(\xgate^b\zgate^a) \zgate\ket{\psi}  \mbox{ (update function) } f_\zgate(a,b)=(a,b)\\
&\hgate (\xgate^b\xgate^a)\zgate=(\xgate^b\xgate^a) \hgate\ket{\psi}  \mbox{ (update function) } f_\hgate(a,b)=(b,a)\\
&\pgate (\xgate^b\xgate^a)\zgate=(\xgate^b\xgate^a) \pgate\ket{\psi}  \mbox{ (update function) } f_ \pgate(a,b)=(a,a\oplus b)\\
&\cnot(\xgate^{a_1}\zgate^{b_1}\otimes \xgate^{a_2}\zgate^{b_2})\ket{\psi}=(\xgate^{b_1}\zgate^{a_1\oplus a_2}\otimes \xgate^{b_1\oplus b_2} \zgate^{ b_2})\cnot(\ket{\psi}) \mbox{ (update function) }\\
& f_{\cnot}(a_1,b_1,a_2,b_2)=(a_1\oplus a_2,b_1,a_2, b_1\oplus b_2).
\end{aligned}
\end{equation*}


\section{Definitions}
\label{sec:definitions}
In this section, we provide a definition of perfectly equivalent quantum circuits (see \Cref{sec:perfectly-equivalent}), and  define our notion of quantum indistinguishability obfuscation for equivalent circuits (\Cref{sec:defs-qiO}). At the end of the section, we also make an observation about the existence of inefficient quantum indistinguishability obfuscation. Note that in~\Cref{sec:quantum:iO:approx:circuits}, we present our alternative definition for quantum indistinguishability obfuscation, applicable to the case where the circuits are  approximately equivalent.

\subsection{Perfectly Equivalent Quantum Circuits}
\label{sec:perfectly-equivalent}

\begin{definition}({\em Perfectly Equivalent Quantum Circuits}):
\label{def:qec}
Let $C_{q_0}$ and $C_{q_1}$ be two $n$-qubit quantum circuits. We say $C_{q_0}$ and $C_{q_1}$ are \emph{perfectly equivalent}  if
$$||C_{q_0}-C_{q_1}||_\diamond=0.$$
\end{definition}

\subsection{Indistinguishability Obfuscation for Quantum Circuits}
\label{sec:defs-qiO}
\begin{definition} ({\em Quantum Indistinguishability Obfuscation for Perfectly Equivalent Quantum Circuits}):
\label{def:QiO}
Let $\mathcal{C}_Q$ be a polynomial-time family of reversible quantum circuits. For $n\in\mathbb{N}$, let $C_{q^n}$ be the circuits in $\mathcal{C}_Q$ of input length $n.$
A  polynomial-time quantum algorithm for~$\mathcal{C}_Q$ is a  \emph{Computational/Statistical/Perfect } \emph{quantum indistinguishability obfuscator} ($qi\mathcal{O}$) if the following conditions hold:
\begin{enumerate}

\item {\tt Functionality:} There exists a negligible function ${\tt negl}(n)$ such that for every $C_q\in C_{q^n}$
$$(\ket{\phi}, C_q^\prime)\leftarrow qi\mathcal{O}(C_q)  \;  \mbox{ and }\;   ||C_q^\prime(\ket{\phi}, \cdot)-C_{q}(\cdot) ||_\diamond=0.$$
Where $\ket{\phi}$ is an $\ell$-qubit state, the circuits $C_q$ and $C_q^\prime$ are of type $(n,n)$ and $(m,n)$ respectively ($m= \ell +n$).\footnote{A circuit is of type $(i,j)$ if it maps $i$ qubits to $j$ qubits.}

\item  {\tt Polynomial Slowdown:} There exists a polynomial $p(n)$ such that for any $C_{q}\in C_{q^n},$
\begin{itemize}
\item  $\ell\leq p(|C_{q}|)$
\item $m \leq  p(|C_{q}|)$
\item $|C_{q}^\prime| \leq p(|C_{q}|).$
\end{itemize}

\item {\tt Computational/Statistical/Perfect Indistinguishability:} For any two perfectly equivalent quantum circuits $C_{q_1},C_{q_2}\in C_{q^n},$ of the same size, the two distributions $qi\mathcal{O}(C_{q_1})$ and $qi\mathcal{O}(C_{q_2})$ are (respectively) computationally/statistically/perfectly indistinguishable.			
\end{enumerate}
\end{definition}

\begin{remark}\label{re:ktime}
A subtlety that is specific to the quantum case is that \Cref{def:QiO} only requires that $(\ket{\phi}, C_q^\prime)$ enable a \emph{single} evaluation of $C_q$.  We could instead require a $k$-time functionality, which can be easily achieved by executing the single-evaluation scheme $k$ times in parallel. This justifies our focus here on the single-evaluation scheme.
\end{remark}

\begin{note}
\label{note:differences-AF16}
 As described in \Cref{sec:obf-quantum}, our \Cref{def:QiO} differs from \cite{AF16arxiv} as it requires security only in the case of equivalent quantum circuits (see \Cref{def:aQiO} for a definition that addresses this). Compared to~\cite{AF16arxiv}, we note that in this work we focus on unitary circuits only.\footnote{This is without loss of generality, since a $qi\mathcal{O}$ for a generalized quantum circuit can be obtained from a $qi\mathcal{O}$ for a reversible version of the circuit, followed by a trace-out operation (see~\Cref{sec:quantum-algo}).} Another difference is that the notion of indistinguishability (computational or statistical) in \cite{AF16arxiv} is more generous than ours, since it allows a finite number of inputs that violate the indistinguishability inequality. Since our work focuses on \emph{possibility} of obfuscations, our choice leads to the strongest results; equally, since \cite{AF16arxiv} focuses on impossibility, their results are strongest in their model.
We also note that that \cite{AF16arxiv} defines the efficiency of the obfuscator in terms of the number of qubits. We believe that our definition, which bounds the size of the output of the obfucation by a polynomial in the \emph{size} of the input circuit,  is more appropriate\footnote{It would be unreasonable to allow an obfuscator that outputs a circuit on $n$ qubits, but of depth super-polynomial in~$n$.} and follows the lines of the classical definitions.
 As far as we are aware, further differences in our definition are purely a choice of style. For instance, we do not include an \emph{interpreter} as in \cite{AF16arxiv}, but instead we let the obfuscator output a quantum circuit together with a quantum state; we chose this presentation since it provides a clear separation between the quantum circuit output by the $qi\mathcal{O}$ and the ``quantum advice state''.
\end{note}

\subsubsection{Inefficient Quantum Indistinguishability Obfuscators Exist}
Finally, we show a simple extension of a result in \cite{BGI+12}, which shows that  if we relax the requirement that the obfuscator be efficient, then information-theoretic indistinguishability obfuscation exists.

 \begin{claim}
 \label{claim:inefficient-qiO}
  Inefficient indistinguishability obfuscators exist for all circuits.
 \end{claim}
 \begin{proof}
 Let $qi\mathcal{O}(C)_q$  be the lexicographically first circuit of size $|C_q|$ that computes the same quantum map as $C_q$.
 \end{proof}


\section{Quantum Indistinguishability Obfuscation for Clifford Circuits}
\label{QiO:Clifford-Circuits}

Here, we show how to construct $qi\mathcal{O}$ for Clifford circuits with respect to definition \Cref{def:QiO}. The first construction (\Cref{sec:Clifford-iO-canonical}) is based on a canonical form, and the second is based on gate teleportation (\Cref{sec:Clifford-iO-teleportaion}).

\subsection{$qi\mathcal{O}$ for Clifford Circuits via a Canonical Form}
\label{sec:Clifford-iO-canonical}
Aaronson and Gottesman developped a polynomial-time algorithm that takes a Clifford circuit $C_q$ and outputs its canonical form (see~\cite{AG04}, section~VI), which is invariant for any two equivalent $n$-qubit circuits\footnote{Their algorithm outputs a canonical form (unique form) provided it runs on the standard initial tableau see pages 8-10 of \cite{AG04}.}.  Moreover the size of the canonical form remains polynomial in the size of the input circuit. Based on this canonical form, we define a $qi\mathcal{O}$ in \Cref{QiO:Canonical-Clifford}.
\begin{algorithm}[]
   \caption{$qi\mathcal{O}$-Canonical}
   \label{QiO:Canonical-Clifford}
  \begin{itemize}
\item Input: An $n$-qubit Clifford Circuit $C_q.$
 \begin{enumerate}
  \item   Using  the Aaronson and Gottesman algorithm~\cite{AG04}, compute the canonical form of $C_q$
   		\begin{equation*}
 		 C_q^\prime \xleftarrow{\mbox{canonical form}}C_q
 		 \end{equation*}
  \item Let $\ket{\phi}$ be an empty register.
  \item Output $\left(\ket{\phi},C_q^\prime \right)$.
  \end{enumerate}
  \end{itemize}
\end{algorithm}

 \begin{lemma}
\Cref{QiO:Canonical-Clifford} is a Perfect Quantum Indistinguishability Obfuscation for all Clifford Circuits.
\end{lemma}

\begin{proof}
 We have to show that \Cref{QiO:Canonical-Clifford} satisfies the definition of a perfect quantum indistinguishability obfuscation (\Cref{def:QiO}) for all Clifford circuits.
 \begin{enumerate}
\item {\tt Functionality:} Since $\ket{\phi}$ is an empty register, it is a 0 qubit state ($\ell=0.$) The circuit $C_q^\prime$ is the canonical form of $C_q,$ therefore, it is also of type $(n,n)$ and has the same functionality as $C_q.$ We have $ ||C_q^\prime(\ket{\phi}, \cdot)-C_{q}(\cdot)||_\diamond =0\leq {\tt negl}(n)$ for any negligible function ${\tt negl}(n).$
\item  {\tt Polynomial Slowdown:} Note $C_q^\prime$ is constructed using Aaronson and Gottesman  algorithm~\cite{AG04}. Therefore, there exists a polynomial $q(\cdot)$ such that $|C_q^\prime|\leq q(|C_q|).$ Let $p(n)=q(n)+n$ then we clearly have
\begin{itemize}
\item $\ell \leq p(|C_{q}|)$
\item $n \leq  p(|C_{q}|)$
\item $|C_{q}^\prime| \leq p(|C_{q}|).$
\end{itemize}
\item {\tt Perfectly Indistinguishability:} Let  $C_{q_1},C_{q_2}\in C_{q^n},$ be any two equivalent Clifford circuits of the same size.  Let
$$(\ket{\phi_1}, C_{q_1}^\prime)\leftarrow qi\mathcal{O}\mbox{-Canonical}(C_{q_1}) \mbox{ and } (\ket{\phi_2}, C_{q_2}^\prime)\leftarrow qi\mathcal{O}\mbox{-Canonical}(C_{q_2}).$$
Since the canonical form of any two equivalent Clifford circuits are exactly the same, we have  $C_{q_1^\prime}=C_{q_2^\prime}.$ Moreover both $\ket{\phi_1}$ and $\ket{\phi_2}$ are empty registers we have $\ket{\phi_1}=\ket{\phi_2}.$ Then we have $qi\mathcal{O}\mbox{-Canonical}(C_{q_1})=qi\mathcal{O}\mbox{-Canonical}(C_{q_2}).$
Therefore, \Cref{QiO:Canonical-Clifford} is a perfect quantum indistinguishability obfuscation for all Clifford circuits.\qedhere
 \end{enumerate}
\end{proof}

\subsection{$qi\mathcal{O}$ for Clifford Circuits via Gate Teleportation}
\label{sec:Clifford-iO-teleportaion}
In this section, we  show how gate teleportation (see \Cref{algo:gate-teleport}) can be used to construct a quantum indistinguishability obfuscation for Clifford circuits. Our construction, given in \Cref{QiO:Clifford-teleportation}, relies on the existence of a quantum-secure~$i\mathcal{O}$ for classical circuits; however, upon closer inspection,
   our construction relies on the assumption that a quantum-secure classical $i\mathcal{O}$ exists for a very specific class of classical circuits\footnote{Circuits that compute update functions for Clifford circuits, see \Cref{update function}.}. In fact, it is easy to construct a perfectly secure $i\mathcal{O}$ for this class of circuits: like Clifford circuits, the circuits that compute the update functions also have a canonical form. Then the $i\mathcal{O}$ takes as input a Clifford circuit and outputs a canonical form of a classical circuit that computes the update function for~$C_q.$ The $i\mathcal{O}$ is described formally in \Cref{alg:classical-iO-Clifford}.

\begin{algorithm}[h!]
   \caption{$qi\mathcal{O}$ via Gate Teleportation for Clifford}
   \label{QiO:Clifford-teleportation}
  \begin{itemize}
  \item Input: An $n$-qubit Clifford Circuit $C_q.$
  \begin{enumerate}
  \item Prepare a tensor product of $n$  Bell states: $\ket{\beta^{ 2n}}=\ket{\beta_{00}}\otimes \cdots \otimes \ket{\beta_{00}}.$
  \item Apply the circuit $C_q$ on the right-most $n$ qubits to obtain a system $\ket{\phi}$:
  										 $$\ket{\phi}=(\igate_n\otimes C_q) \ket{\beta^{2n}}.$$
 \item Compute a classical circuit $C$ that computes the update function $F_{C_q}.$  The classical circuit $C$ can be computed in polynomial-time by \Cref{lemma: iO-clifford-functions}.
\item Set $C^\prime\leftarrow i\mathcal{O}(C)$, where $i\mathcal{O}(C)$ is a perfectly secure indistinguishability obfuscation defined in \Cref{sec: iO-clifford-functions}.													
  \item Description of the circuit  $C_q^\prime:$
 \begin{enumerate}
     \item  Perform a general Bell measurement on the leftmost $2n$-qubits on the system $\ket{\phi}\otimes \ket{\psi}$, where $\ket{\phi}$ is an auxiliary state and $\ket{\psi}$ is an input state. Obtain classical bits $(a_1,b_1\ldots,a_n,b_n)$ and the state
     \begin{equation}
     \label{bld:eq1:QiO-Clifford}
      C_q(\xgate^{\otimes_{i=1}^{n} b_{i}} \cdot \zgate^{\otimes_{i=1}^{n} a_{i}})\ket{\psi}.
      \end{equation}
        \item Compute the correction bits
       \begin{equation}
      \label{bld:eq2:QiO-Clifford}
    (a_1^\prime, b_1^\prime,\ldots, a_n^\prime, b_n^\prime)=C^\prime(a_1,b_1\ldots,a_n,b_n).
  \end{equation}
     \item Using the above, the correction unitary is
      $U^\prime=(\xgate^{\otimes_{i=1}^{n} b_{i}^\prime} \cdot \zgate^{\otimes_{i=1}^{n} a_{i}^\prime}).$
     \item Apply  $U^\prime$ to the system $C_q(\xgate^{\otimes_{i=1}^{n} b_{i}} \cdot \zgate^{\otimes_{i=1}^{n} a_{i}})\ket{\psi}$ to obtain the state $C_q(\ket{\psi}).$
     \end{enumerate}
  \item Output $\left(\ket{\phi},C_q^\prime \right).$
  \end{enumerate}
  \end{itemize}
\end{algorithm}
\pagebreak

\begin{theorem}\label{th:qio:cliff}
 \Cref{QiO:Clifford-teleportation} is a perfect quantum indistinguishability obfuscation for all Clifford Circuits.
\end{theorem}
\begin{proof} We have to show that \Cref{QiO:Clifford-teleportation} satisfies \Cref{def:QiO}.
\begin{enumerate}
\item  {\tt Functionality:} Let $C_q$ be an $n$-qubit Clifford Circuit  and $\left(\ket{\phi},C_q^\prime \right)$ be the output of the \Cref{QiO:Clifford-teleportation} on input $C_q.$ On input $(\igate_n\otimes C_q) \ket{\beta^{2n}}$ and $\ket{\psi}$ the circuit $C_q^\prime$ outputs the state $C_q(\ket{\psi})$ (this follows from the principle of gate teleportation). Therefore $C_q^\prime(\ket{\phi}, \cdot)=C_q(\psi)$, which implies that
$||C_q^\prime(\ket{\phi}, \cdot)-C_{q}(\cdot) ||_\diamond=0\leq {\tt negl}(n)$ for any negligible function ${\tt negl}(n).$

\item {\tt Polynomial Slowdown:} Note $\ell=2n$ (the number of qubits in $\ket{\phi}$) and $C_q^\prime$ is a circuit of type $(3n,n).$ The size of the circuit $|C_q^\prime|= |i\mathcal{O}(C)|+|\mbox{Bell measurement}|$, where $C$ is the classical circuit that computes the update function corresponding to $C_q.$ The size of a Bell measurement circuit for an $O(n)$-qubit state is $O(n).$ Therefore, there exists a polynomial $q(|C|)$ such that $|\mbox{Bell measurement}|\leq q(|C|)|.$ The size of $|i\mathcal{O}(C)|$ is at most $r(|C|)$ for some polynomial $r(\cdot)$ (\Cref{lemma: iO-clifford-functions}). Further, the size of $|C|$ is at most $s(|C_q|)$ for some polynomial $|C_q|$ (\Cref{lemma: iO-clifford-functions}).  By setting $p(|C|)=2|C_q|+q(|C|)+r(s(|C|)),$ we have

\begin{itemize}
\item $\ell\leq p(|C|)$
\item $m=3n \leq p(|C|)$
\item $|C_q^\prime|= |i\mathcal{O}(C)|+|\mbox{Bell measurement}|\leq p(|C|).$
\end{itemize}

\item {\tt Perfect Indistinguishability:} Let $C_{q_1}$ and $C_{q_2}$ be two $n$-qubit equivalent Clifford circuits of the same size. Let $\left(\ket{\phi_1},C_{q_1}^\prime \right)$ and $\left(\ket{\phi_2},C_{q_2}^\prime \right)$ be the outputs of \Cref{QiO:Clifford-teleportation} on inputs $C_{q_1}$ and $C_{q_2}$ respectively.
Since  $C_{q_1}(\ket{\tau})=C_{q_2}(\ket{\tau})$ for every quantum state $\ket{\tau}$ we have,

\begin{equation}
\ket{\phi_1}=(I\otimes C_{q_1}) \ket{\beta^{2n}}=(I\otimes C_{q_2}) \ket{\beta^{2n}}=\ket{\phi_2}.
\end{equation}
 The update functions for any two equivalent Clifford circuits are equivalent (\Cref{lem:clifford-functions}), further the classical $i\mathcal{O}$ that obfuscates the update functions (circuits) is perfectly indistinguishable for any two equivalent Clifford circuits (not necessarily of the same size) (\Cref{lemma: iO-clifford-functions}). Therefore, $C_{q_1}^\prime$ and $C_{q_2}^\prime$ are perfectly indistinguishable. \qedhere
\end{enumerate}
\end{proof}

\begin{lemma}\label{lem:clifford-functions}
Let $C_{q_1}$ and $C_{q_2}$ be two equivalent $n$-qubit Clifford circuits. Then their corresponding update functions are also equivalent.
\end{lemma}

\begin{proof}
Let $F_{C_{q_1}}$ and $F_{C_{q_2}}$ be the update functions for two $n$-qubit Clifford circuits  $C_{q_1}$ and $C_{q_2}$ respectively. Suppose $F_{C_{q_1}}\neq F_{C_{q_2}},$  then there must exist at least one binary string ${\bf s}=a_1b_1 \ldots  a_n b_n \in \{0,1\}^{2n} $ such that
\begin{equation}
\label{bld:equ1:QiO-Clifford}
F_{C_{q_1}}({\bf s})\neq F_{C_{q_2}}({\bf s})
\end{equation}

Since $C_{q_1}$ and $C_{q_2}$  are equivalent circuit we must have that for every quantum state $\ket{\psi}$:
\begin{equation}
\label{bld:equ2:QiO-Clifford}
C_{q_1} (\xgate^{\otimes_{i=1}^{n} b_{i}} \cdot \zgate^{\otimes_{i=1}^{n}a_{i}})\ket{\psi}=C_{q_2} (\xgate^{\otimes_{i=1}^{n} b_{i}} \cdot \zgate^{\otimes_{i=1}^{n}a_{i}})\ket{\psi}
\end{equation}
Let $F_{C_{q_1}}({\bf s})=(a_1^\prime b_1^\prime,\dots a_n^\prime b_n^\prime)$ and $F_{C_{q_2}}({\bf s})=(d_1^\prime e_1^\prime, \dots, d_n^\prime e_n^\prime),$ then we can rewrite \Cref{bld:equ2:QiO-Clifford} as

\begin{equation}
\label{bld:eq3:QiO-Clifford}
 (\xgate^{\otimes_{i=1}^{n} b_{i}^\prime} \cdot \zgate^{\otimes_{i=1}^{n}a_{i}^\prime})C_{q_1} (\ket{\psi})=(\xgate^{\otimes_{i=1}^{n} e_{i}^\prime} \cdot \zgate^{\otimes_{i=1}^{n}d_{i}^\prime})C_{q_2}(\ket{\psi}).
\end{equation}

We can replace $C_{q_2} (\ket{\psi})$ with $C_{q_1} (\ket{\psi})$ in \Cref{bld:eq3:QiO-Clifford}
\begin{equation}
\label{bld:eq4:QiO-Clifford}
(\xgate^{\otimes_{i=1}^{n} b_{i}^\prime} \cdot \zgate^{\otimes_{i=1}^{n}a_{i}^\prime})C_{q_1} (\ket{\psi})=(\xgate^{\otimes_{i=1}^{n} e_{i}^\prime} \cdot \zgate^{\otimes_{i=1}^{n}d_{i}^\prime})C_{q_1}(\ket{\psi}).
\end{equation}

Now if there exists a $j$ such that $a_j^\prime \neq d_j^\prime$ or  $b_j^\prime \neq e_j^\prime,$ then \Cref{bld:eq4:QiO-Clifford} does not hold.  This contradicts the assumption that $C_{q_1}$ and $C_{q_2}$ are equivalent Clifford circuits. Therefore $F_{C_{q_1}}$ and $F_{C_{q_2}}$ are equivalent functions.\footnote{The \Cref{bld:eq3:QiO-Clifford}  is derived from the assumption that $C_{q_1}$ and $C_{q_2}$ are equivalent Clifford circuits.}
\end{proof}

\subsubsection{Indistinguishability Obfuscator for Clifford: Update Functions}
 \label{sec: iO-clifford-functions}
 Here, we describe a perfect indistinguishability obfuscator $i\mathcal{O}$ for the update functions corresponding to the Clifford circuits. The algorithm takes an $n$-qubit Clifford circuit $C_q$ and output a classical circuit $C$ that computes the update function~$F_{C_q}.$ The circuit $C$ is invariant for any two equivalent Clifford circuits. The main idea here is to compute the canonical form for the Clifford circuit, and then compute the update function for the canonical form.

\begin{algorithm}[]
\caption{$i\mathcal{O}$ for Clifford: Update Functions.}
\label{alg:classical-iO-Clifford}
  \begin{enumerate}
  \item Compute the canonical form $C_q$ using the algorithm presented in ~\cite{AG04} (section VI). Denote the canonical form as $\widehat{C}_q.$
  \item Let $g_1,g_2, \dots, g_m$ be a topological ordering of the gates in  $\widehat{C}_q,$ where $m=|{\widehat{C}_q}|.$
 \item Construct the classical circuit $\hat{C}$ that computes the update function $F_{\widehat{C}_q}$ as follows. For $i=1$ to $m$, implement the update rule for each gate $g_i$ (\Cref{update function}).
  \item Output the classical circuit $\hat{C}.$
  \end{enumerate}
\end{algorithm}

\begin{lemma}\label{lemma: iO-clifford-functions}
\Cref{alg:classical-iO-Clifford} is a perfect classical indistinguishability obfuscator for the Clifford update functions.
\end{lemma}

\begin{proof} We have to show that \Cref{alg:classical-iO-Clifford} satisfies \Cref{def:quantum-secureiO}.
\begin{enumerate}
\item {\tt Functionality:} Let $C_q$ be an $n$-qubit Clifford circuit and $C$ be the circuit that computes $F_{C_q}$ (the update function for $C_q).$ Let $\widehat{C}_q$ be the canonical form of  $C_q.$ Since $C_q$ and ${C_q}^\prime$ are equivalent circuits, it follows from \cref{lem:clifford-functions} that  $F_{C_q}$ and $F_{C_q}^\prime$ are equivalent functions, therefore any circuit that computes $F_{C_q}^\prime$ also computes $F_{C_q}.$ From the construction above (\Cref{alg:classical-iO-Clifford}) it follows that the circuit $\hat{C}$ computes $F_{C_q}^\prime$ therefore  $\hat{C}(x)= C(x)$ for all inputs~$x.$

\item  {\tt Polynomial Slowdown:}  For each gate $g_i$ in $\widehat{C}_q$, the classical circuit $\widehat{C}$ has to implement one of the following operations (\Cref{update function}):
 \begin{itemize}
 \label{cost}
 \item[] $(a,b)\xrightarrow{\xgate, \zgate}(a,b).$
 \item[]  $(a,b)\xrightarrow{\hgate}(b,a)$         (one swap).
 \item[]  $(a,b)\xrightarrow{\pgate}(a,a\oplus b)$      (one $\oplus$ operation).
 \item[]  $(a_1,b_1,a_2,b_2)\xrightarrow{\mbox{\cnot}}(a_1\oplus a_2,b_1,a_2, b_1\oplus b_2)$ (two $\oplus$ operations).
 \end{itemize}
 Therefore the size of $|\widehat{C}|$ can be at most be $O(|\widehat{C}_q|).$ The $\widehat{C}_q$ is a canonical form of $C_q$ and of at most $q(|C_q|)$ for some polynomial $q(\cdot)$ \cite{AG04}.  Therefore, there exists a polynomial $p(\cdot)$ such that $|\widehat{C}|\leq p(|C_q|).$

\item {\tt Perfectly Indistinguishability:}  Let $C_{q_1}, C_{q_2}$ be two equivalent $n$-qubit Clifford circuits (not necessarily of the same size) and $\widehat{C_{q}}$ be their canonical form. Note that the output of  \cref{alg:classical-iO-Clifford} only depends on the canonical form of the input Clifford circuit. Since, $C_{q_1}$ and $C_{q_2}$ have the same canonical form we have
												$$\hat{C}\leftarrow\Cref{alg:classical-iO-Clifford}(\widehat{C_{q}})=\Cref{alg:classical-iO-Clifford}(C_{q_1})=\Cref{alg:classical-iO-Clifford}(C_{q_2}),$$
Therefore, \Cref{alg:classical-iO-Clifford} is a perfect indistinguishability obfuscation for the Clifford update functions.\qedhere
\end{enumerate}
\end{proof}
\begin{remark}
\label{re:iO:sizes}
What is convenient about the $i\mathcal{O}$ of \Cref{alg:classical-iO-Clifford} is that it works for any two equivalent Clifford circuits (regardless of their relative sizes) (see \Cref{lem:clifford-functions}). However, we can use any perfectly secure $i\mathcal{O}$ in our construction with some care. Suppose $i\mathcal{O}$ is some perfectly secure indistinguishability obfuscator for classical circuits (for Clifford update functions) of the same size. Suppose we want to obfuscate an update circuit corresponding to some Clifford~$C_q$. The classical circuit~$C$ is constructed by going through each gate in~$C_{q}.$ Some gates are more costly than others (for \emph{e.g.}, $\cnot$ vs. $\zgate,$ see proof of \Cref{th:qio:cliff} or \Cref{update function}). Since we assume all Clifford circuits are of the same size, we can obtain an upper bound on all the classical circuits (for the update functions) by replacing each gate in $C_q$ with the most costly gate and then computing the classical circuit for the resulting quantum gate. Now suppose $m$ is the upper bound on the size of classical circuits, then for any circuit $C_q,$ we first calculate the circuit $C$ that computes $F_{C_q}$ and then pad $C$ with $m-|C|$ identity gates. This will ensure that if $|C_{q_1}|=|C_{q_2}|,$ then $|C_1|=|C_2|.$
\end{remark}


\section{Obfuscating Beyond Clifford Circuits}
\label{QiO:Clifford+T:family}
In this section, we extend the gate teleportation technique to show how we can construct $qi\mathcal{O}$ for \emph{any} quantum circuit. Our construction is efficient as long as the circuit has $\tgate$-count at most logarithmic  in the circuit size.  For the sake of simplicity, we first construct a $qi\mathcal{O}$ for an arbitrary 1-qubit quantum circuit (\Cref{sec:1-qubit}), then  extend the 1-qubit construction to any $n$-qubit quantum circuit (\Cref{sec:n-qubit-circuits}).

We first start with some general observations on quantum circuits which are relevant to this section. Consider the application of the $\tgate$-gate on an encrypted system using the quantum one-time pad. The following equation relates  the $\tgate$-gate to the  $\xgate$- and  $\zgate$-gates:
\begin{equation}
\label{eq:pgate}
\tgate \xgate^b \zgate^a=\xgate^b \zgate^{a\oplus b}\pgate^b \tgate\,.
\end{equation}
If $b=0,$ then  $\pgate^b$   is the identity; otherwise we have a $\pgate$-gate correction. This is undesirable as $\pgate$ does not commute with $\xgate$, making the update of the encryption key $(a,b)$ complicated (since it is no longer a tensor product of Paulis). Note that we can write
$\pgate=\left(\frac{1+i}{2}\right) \igate + \left(\frac{1-i}{2}\right)\zgate$,
therefore  \Cref{eq:pgate} can be rewritten as:
\begin{equation}
\label{eq:pgate1}
\tgate \xgate^b \zgate^a=\xgate^b \zgate^{a\oplus b}\left[\left(\frac{1+i}{2}\right) \igate + \left(\frac{1-i}{2}\right)\zgate\right]^b \tgate
\end{equation}
Since  $\left[\left(\frac{1+i}{2}\right) \igate + \left(\frac{1-i}{2}\right)\zgate\right]^b=\left(\frac{1+i}{2}\right) \igate + \left(\frac{1-i}{2}\right)\zgate^b$ for $b \in \{0,1\},$ we can rewrite  \Cref{eq:pgate1} as,
\begin{equation}
\label{eq:pgate2}
\begin{aligned}
\tgate \xgate^b \zgate^a=\xgate^b \zgate^{a\oplus b}\left[\left(\frac{1+i}{2}\right) \igate + \left(\frac{1-i}{2}\right)\zgate^b\right] \tgate \\
=\left[\left(\frac{1+i}{2}\right) \xgate^b \zgate^{a\oplus b} + \left(\frac{1-i}{2}\right)\xgate^b \zgate^a\right] \tgate. \\
\end{aligned}
\end{equation}

It follows from \Cref{eq:pgate2} that for any $a,b\in\{0,1\},$ we can represent $\tgate \xgate^b \zgate^a$ as a linear combination of $\xgate$ and $\zgate.$
\begin{equation}
\label{eq:pgate3}
\begin{aligned}
\tgate \xgate^b \zgate^a= (\alpha_1 \igate +  \alpha_2 \xgate + \alpha_3 \zgate + \alpha_4 \xgate\zgate)\tgate \\
\end{aligned}
\end{equation}
where $\alpha_j \in\left\{0,1, \frac{1+i}{2}, \frac{1-i}{2}\right\},$ for  $j\in[4].$

We further note that for a general $n$-qubit quantum unitary $U$ and $n$-qubit Pauli~$P$, there exists a Clifford $C$ such that $UP\ket{\psi} = CU \ket{\psi}$. This is due to the \emph{Clifford hierarchy} \cite{GC99}. We also mention that if an $n$-qubit Clifford operation is given in matrix form, an efficient procedure exists in order to produce a circuit that executes this Clifford\cite{NWD14}. This is a special case of the general problem of \emph{synthesis} of quantum circuits, which aims to produce quantum circuits, based on an initial description of a unitary operation.

\subsection{Single-Qubit Circuits}
\label{sec:1-qubit}

Here, we show an indistinguishability obfuscation for single-qubit circuits. As previously mentionned, we note that for the single-qubit case, an efficient indistinguishability obfuscation can also be built using the Matsumoto-Amano normal form~\cite{MA08arxiv,GS19arxiv}.  Here, we give an alternate construction based on gate teleportation.
Let $C_q$ be a $1$-qubit circuit we want to obfuscate and $\ket{\psi}$ be the quantum state on which we want to evaluate $C_q.$ Note that we can write any $1$-qubit circuit as a sequence of gates from the set $\{\hgate, \tgate\}$\footnote{The set $\{\hgate,  \tgate\}$ is universal for 1-qubit unitaries \cite{KLM07}.}
$$C_q=(g_{|C_q|},\ldots,g_2,g_1 ),\; g_i\in\{\hgate, \tgate\}\,.$$
For the indistinguishability obfuscation of a single-qubit circuit, we use the gate teleportation protocol (\Cref{algo:gate-teleport}), which leaves us (after the teleportation) with a subsystem of the form $C_q \xgate^b\zgate^a (\ket{\psi})$
\begin{equation}
\label{eq:pgate4}
C_q \xgate^b\zgate^a (\ket{\psi})=(g_{|C_q|},\ldots,g_2,g_1 )\xgate^b\zgate^a (\ket{\psi}),
\end{equation}
and to evaluate the circuit on $\ket{\psi}$, we have to apply a correction unitary. Now suppose we apply the gate $g_1$. We can write the system in \Cref{eq:pgate4} as
\begin{equation}
\begin{aligned}
\label{eq:pgate5}
C_q \xgate^b\zgate^a (\ket{\psi})=(g_{|C_q|},\ldots,g_2)(\alpha_0 \igate +  \alpha_1 \xgate + \alpha_2 \zgate + \alpha_3 \xgate \zgate)g_1 (\ket{\psi})
\end{aligned}
\end{equation}
 where $\alpha_i \in\left\{0,1, \frac{1+i}{2}, \frac{1-i}{2}\right\}.$ Since $\{\igate, \xgate, \zgate, \xgate \zgate\},$ forms a basis, after applying the remaining gates in the sequence  $(g_{|C_q|},\ldots,g_3,g_2),$ we can write \Cref{eq:pgate5} as
\begin{equation}
\begin{aligned}
\label{eq:pgate8}
C_q \xgate^b\zgate^a (\ket{\psi})=(\beta_1 \igate +  \beta_2 \xgate + \beta_3 \zgate + \beta_4\xgate \zgate)(g_{|C_q|},\ldots,g_2,g_1 )(\ket{\psi})
\end{aligned}
\end{equation}
where each $\beta_i\in\mathbb{C}$ and is computed by multiplying and adding numbers from the set $\{0,1,\frac{1+i}{2},\frac{1-i}{2}\}.$ We show in \Cref{coeff:size} that the size of the coefficients~$\beta_i$  grows at most as a polynomial in the number of $\tgate$-gates. Therefore it follows from \Cref{eq:pgate8} that the update function for any $1$-qubit circuit $C_q$ can be defined as the following map,
\begin{equation*}
F_{C_q}:\{0,1\}^2\rightarrow \mathbb{C}^4,\; (a,b)\mapsto (\beta_1, \beta_2,\beta_3,\beta_4),
\end{equation*}
and is in one-to-one correspondence with the correction unitary $\beta_1 \igate +  \beta_2 \xgate + \beta_3 \zgate + \beta_4\xgate \zgate.$
As indicated, our construction for 1-qubit circuits is nearly the same as the gate teleportation scheme for Clifford circuits (\Cref{QiO:Clifford-teleportation}). The proof that this is a $qi\mathcal{O}$ scheme is also very similar to the proof for the Clifford construction (\Cref{sec:Clifford-iO-teleportaion}); we thus omit the formal proof here (it can also be seen as a special case of the proof of \Cref{theorem:main}). Some subtleties, however are addressed below: the equivalence of the update functions (\Cref{lemma:1qubit}) and the circuit synthesis (\Cref{lem:1-qubit-synthesis}).

\begin{lemma}
\label{lemma:1qubit}
Let $C_{q_1}$ and $C_{q_2}$ be two equivalent $1$-qubit circuits. Then their corresponding update functions in the gate teleportation protocol are also equivalent.
\end{lemma}

\begin{proof}
Suppose $F_{C_{q_1}}(a,b)=(\beta_1,  \beta_2 , \beta_3 , \beta_4,)$ and $F_{C_{q_2}}(a,b)=(\gamma_1, \gamma_2, \gamma_3, \gamma_4)$ are the corresponding update functions for $a,b\in\{0,1\}.$ Since $C_{q_1}$ and $C_{q_2}$ are equivalent circuits, for every quantum state $\ket{\psi}$ and any $a,b \in\{0,1\},$

$$C_{q_1} \xgate^b \zgate^a(\ket{\psi})=C_{q_2} \xgate^b \zgate^a(\ket{\psi})$$
$$\Leftrightarrow (\beta_1 \igate +  \beta_2 \xgate + \beta_3 \zgate + \beta_4\xgate \zgate)C_{q_1}(\ket{\psi})= (\gamma_1 \igate + \gamma_2 \xgate + \gamma_3 \zgate + \gamma_4\xgate \zgate)C_{q_2}(\ket{\psi})$$
$$\Leftrightarrow (\beta_1 \igate +  \beta_2 \xgate + \beta_3 \zgate + \beta_4\xgate \zgate)C_{q_1}(\ket{\psi})= (\gamma_1 \igate + \gamma_2 \xgate + \gamma_3 \zgate + \gamma_4\xgate \zgate)C_{q_1}(\ket{\psi})$$
$$\Leftrightarrow ((\beta_1- \gamma_1)\igate +  (\beta_2 -\gamma_2) \xgate + (\beta_3 - \gamma_3)\zgate + (\beta_4 -\gamma_4)\xgate \zgate)C_{q_1}(\ket{\psi})= {\bf 0}.$$
$$\Rightarrow (\beta_1- \gamma_1)\igate +  (\beta_2 -\gamma_2) \xgate + (\beta_3 -\gamma_3)\zgate + (\beta_4 -\gamma_4)\xgate \zgate= {\bf 0}.$$
$$\Rightarrow \beta_1= \gamma_1,\,  \beta_2 =\gamma_2,\, \beta_3 =\gamma_3,\, \beta_4 =\gamma_4.$$

Therefore, $F_{C_{q_1}}$ and $F_{C_{q_2}}$ are equivalent functions.
\end{proof}

We note that, on top of being equal, the circuits that compute the update functions $F_{C_{q_1}},$  $F_{C_{q_2}}$ can be assumed to be of the same size. This follows by an argument very similar to the one in \Cref{re:iO:sizes}.\\

\begin{lemma}
\label{lem:1-qubit-synthesis}
Based on the classical $i\mathcal{O}$ that computes the coefficients in \Cref{eq:pgate8}, it is possible to build a quantum circuit that performs the correction efficiently.
\end{lemma}

\begin{proof}
Given a $2 \times 2$ unitary matrix that represents a Clifford operation as in \Cref{eq:pgate8}, it is simple to efficiently derive the Clifford circuit that implements the unitary. This is a special case of the general efficient synthesis for Clifford circuits as presented in~\cite{NWD14}. 
\end{proof}

\subsection{$qi\mathcal{O}$ via Gate Teleportation for all Quantum Circuits}
\label{sec:n-qubit-circuits}

In this section, we construct a $qi\mathcal{O}$ for all quantum circuits. The construction is efficient whenever the number of $\tgate$-gates is at most logarithmic in the circuit size (see \Cref{QiO:qcircuit-teleportation}). The reason for this limitation is that the update function blows up once the number of $\tgate$-gates is greater than logarithmic in the circuit size. The construction is very similar to the gate teleportation for Clifford circuits (\Cref{sec:Clifford-iO-teleportaion}) and assumes the existence of a quantum-secure $i\mathcal{O}$ for classical circuits.

\begin{algorithm}[]
   \caption{$qi\mathcal{O}$ via Gate Teleportation for Quantum Circuits}
   \label{QiO:qcircuit-teleportation}
  \begin{itemize}
  \item Input: A $n$-qubit quantum Circuit $C_q$ with $\tgate$-count $\in O(\log(|C_q|)).$
  \begin{enumerate}
 \item Prepare a tensor product of $n$  Bell states: $\ket{\beta^{ 2n}}=\ket{\beta_{00}}\otimes \cdots \otimes \ket{\beta_{00}}.$
  \item Apply the circuit $C_q$ on the right-most $n$ qubits to obtain a system $\ket{\phi}$:
  										 $$\ket{\phi}=(\igate_n\otimes C_q) \ket{\beta^{2n}}.$$
  \item Set $\hat{C}\leftarrow i\mathcal{O}(C).$ Where $C$ is a circuit that computes the update function $F_{C_q}$ as in \Cref{sec:sizeofupdate}. Note  the size of $C$ is at most a  polynomial in $|C_q|$ (\Cref{lemma:nqubit:cost}).	
    \item  Description of the circuit  $C_q^\prime:$
    \begin{enumerate}
     \item Perform a general Bell measurement on the leftmost $2n$-qubits on the system $\ket{\phi}\otimes \ket{\psi}$, where $\ket{\phi}$ is an auxiliary state and $\ket{\psi}$ is an input state. Obtain classical bits $(a_1,b_1\ldots,a_n,b_n)$ and the state
 $$C_q(\xgate^{\otimes_{i=1}^{n} b_{i}} \cdot \zgate^{\otimes_{i=1}^{n} a_{i}})\ket{\psi}.$$
     \item Compute the correction using the obfuscated circuit
  $$((\beta_1, {\bf s}_1),\ldots, (\beta_n, {\bf s}_k))=\hat{C}(a_1,b_1\ldots,a_n,b_n).$$
     \item Using the above, the correction unitary is $$U_{F_{C_q}}=\sum_{i=1}^{4^k} \beta_i \xgate^{b_{i_1}} \zgate^{a_{i_1}}\otimes \cdots \otimes \xgate^{b{i_n}} \zgate^{a_{i_n}}.$$
     Compute a quantum circuit that applies $U_{F_{C_q}}$, using the circuit synthesis method of \cite{NWD14}.
     \item Apply the quantum circuit for $U_{F_{C_q}}$ to the system $C_q(\xgate^{\otimes_{i=1}^{n} b_{i}} \cdot \zgate^{\otimes_{i=1}^{n} a_{i}})\ket{\psi}$ to obtain the state $C_q(\ket{\psi}).$
     \end{enumerate}
  \end{enumerate}
  \end{itemize}
\end{algorithm}

We are now ready to present our main theorem (\Cref{theorem:main}). For ease of presentation, the proof relies on three auxiliary lemmas that are presented in the following section: \Cref{lemma:nqubit} (which shows that equivalent circuits have equivalent update functions), \Cref{lem1:mainTh} (which bounds the number of terms of the update function), and \Cref{lemma:nqubit:cost} (which shows that update functions can be computed by a polynomial-size circuits).

\begin{theorem} (Main Theorem)\label{theorem:main}
If $i\mathcal{O}$ is a perfect/statistical/computational quantum-secure indistinguishability obfuscation for classical circuits, then \Cref{QiO:qcircuit-teleportation} is a perfect/statistical/computational quantum indistinguishability obfuscator for any quantum circuit~$C_q$ with $\tgate$-count $\in O(\log|C_q|).$
\end{theorem}

\begin{proof} We have to show that \Cref{QiO:qcircuit-teleportation}  satisfies \Cref{def:QiO}. Throughout  the proof, we assume that the quantum circuits have a logarithmic $\tgate$-count in the circuit size.

\begin{enumerate}
\item  {\tt Functionality:} The proof of functionality follows from the principle of gate teleportation  and is very similar to the proof in \Cref{th:qio:cliff}. Since the Clifford circuit synthesis has perfect correctness\cite{NWD14}, we have
$||C_q^\prime(\ket{\phi}, \cdot)-C_{q}(\cdot) ||_\diamond=0\leq {\tt negl}(n)$ for any negligible function ${\tt negl}(n).$

\item {\tt Polynomial Slowdown:} Note that $\ket{\phi}$ is a $2n$-qubit state and $C_q^\prime$ is of type $(m,n),$ where $m=3n,$ therefore both $\ket{\phi}$ and $m$ have size  in $O(|C_q|).$ The size of $C_q^\prime$ is equal to the size of $\hat{C}$ plus the size of the circuit that performs the general Bell measurement ($GBM$) and the size of the circuit that computes the circuit for $U_{F_{C_q}}$.
														 Since the size of $|GBM|$ is in $O(n)$, the size of $|\hat{C}|$ is  polynomial in $|C_q|$ (\Cref{lem1:mainTh}, \Cref{lemma:nqubit:cost}, \Cref{lem:beta:size} and the definition of $i\mathcal{O}$). Moreover, efficient Clifford synthesis implies that the size of the circuit for $U_{F_{C_q}}$ is polynomial~\cite{NWD14}.
 Therefore, there exists a polynomial $p(\cdot)$ such that
\begin{itemize}
\item $|\ket{\phi}|\leq p(|C_q|)$
\item $m\leq p(|C_q|)$
\item  $|C_q^\prime|\leq p(|C_q|).$
\end{itemize}

\item {\tt Perfect/Statistical/Computational Indistinguishability:} Let $C_{q_1}$ and $C_{q_2}$ be two $n$-qubit circuits of the same size. Let $\left(\ket{\phi_1},C_{q_1}^\prime \right)$ and $\left(\ket{\phi_2},C_{q_2}^\prime \right)$ be the outputs of \Cref{QiO:Clifford-teleportation} on inputs $C_{q_1}$ and $C_{q_2}$ respectively.
Since  $C_{q_1}(\ket{\tau})=C_{q_2}(\ket{\tau})$ for every quantum state $\ket{\tau}$ we have, \looseness=-1

\begin{equation}
\ket{\phi_1}=(I\otimes C_{q_1}) \ket{\beta^{2n}}=(I\otimes C_{q_2}) \ket{\beta^{2n}}=\ket{\phi_2},
\end{equation}
The update functions for any two equivalent quantum circuits are equivalent (\Cref{lemma:nqubit}). If the classical $i\mathcal{O}$ that obfuscates the circuits for the update functions  is perfectly/statistically/computationally indistinguishable, then states $C_{q_1}^\prime$ and $C_{q_2}^\prime$ are perfectly/statistically/computationally indistinguishable.\footnote{Note that circuits that compute update functions (for equivalent quantum circuits) may have different sizes. However, that can be managed as discussed in \Cref{re:iO:sizes}.} Therefore, \Cref{QiO:qcircuit-teleportation} is a perfectly/statistically/ computationally indistinguishable quantum obfuscator for the quantum circuits. \qedhere
\end{enumerate}
\end{proof}

\subsection{Equivalent Update Functions}
\label{sec:Update-equivalent}

Here, we provide a generalization of  \Cref{lem:clifford-functions}, applicable to the case of general circuits.

\begin{lemma}
\label{lemma:nqubit}
Let $C_{q_1}$ and $C_{q_2}$ be two equivalent $n$-qubit circuits. Then their corresponding update functions are also equivalent.
\end{lemma}

\begin{proof}
Suppose $C_{q_1}$ and $C_{q_2}$ are two equivalent $n$-qubit quantum circuits, then for any quantum state $\ket{\psi}$ and for any binary string ${\bf r}\in\{0,1\}^{2n}.$
\begin{equation}
\label{equ1:nqubit}
C_{q_1} (\xgate^{a_{i}}\zgate^{b_{i}})^{\otimes_{i=1}^{n}}\ket{\psi} =C_{q_2} (\xgate^{a_{i}}\zgate^{b_{i}})^{\otimes_{i=1}^{n}}\ket{\psi}.
\end{equation}
Then the corresponding update functions are (see \Cref{nqubit:eq0})

\begin{center}
$F_{C_{q_1}}({\bf r})=((\beta_1, {\bf s}_1),\ldots, (\beta_n, {\bf s}_k))$ \\
$F_{C_{q_1}}({\bf r})=((\beta_1^\prime, {\bf s}_1^\prime),\ldots, (\beta_{\ell}^\prime, {\bf s}_{\ell}^\prime))$
\end{center}
where ${\bf s}_i=a_{i_1}b_{i_1}, \ldots, a_{i_n}b_{i_n} \in\{0,1\}^{2n},$  and ${\bf s}_j^\prime=a_{j_1}b_{j_1}, \ldots, a_{j_n}b_{j_n} \in\{0,1\}^{2n}.$
Without loss of generality, we can assume that $\beta_i\neq 0, \beta_j^\prime\neq 0$ for $i\in[k], j\in[\ell].$ Using the update functions, we can rewrite \Cref{equ1:nqubit} as

\begin{equation}\label{equ2:nqubit}
\left(\sum_{i=1}^{4^k} \beta_i \xgate^{b_{i_1}} \zgate^{a_{i_1}}\otimes \cdots \otimes \xgate^{b{i_n}} \zgate^{a_{i_n}}\right)C_{q_1}\ket{\psi} =\left(\sum_{j=1}^{4^{\ell}} \beta_j^\prime \xgate^{b_{j_1}^\prime} \zgate^{a_{j_1}^\prime}\otimes \cdots \otimes \xgate^{b_{j_n}^\prime} \right)C_{q_2}\ket{\psi}.
\end{equation}
Since $C_{q_1}$ and $C_{q_2}$ are equivalent, we can replace $C_{q_2}$ with $C_{q_1}$ in  \Cref{equ2:nqubit}

\begin{equation}\label{equ3:nqubit}
\left(\sum_{i=1}^{4^k} \beta_i \xgate^{b_{i_1}} \zgate^{a_{i_1}}\otimes \cdots \otimes \xgate^{b{i_n}} \zgate^{a_{i_n}}\right)C_{q_1}\ket{\psi} =\left(\sum_{j=1}^{4^{\ell}} \beta_j^\prime \xgate^{b_{j_1}^\prime} \zgate^{a_{j_1}^\prime}\otimes \cdots \otimes \xgate^{b_{j_n}^\prime} \right)C_{q_1}\ket{\psi}
\end{equation}
Then \Cref{equ3:nqubit} can only hold if

\begin{equation}\label{equ4:nqubit}
\left(\sum_{i=1}^{4^k} \beta_i \xgate^{b_{i_1}} \zgate^{a_{i_1}}\otimes \cdots \otimes \xgate^{b{i_n}} \zgate^{a_{i_n}}\right)=\left(\sum_{j=1}^{4^{\ell}} \beta_j^\prime\xgate^{b_{j_1}^\prime} \zgate^{a_{j_1}^\prime}\otimes \cdots \otimes \xgate^{b_{j_n}^\prime} \right)
\end{equation}

Note the update functions $F_{C_{q_1}}$ and $F_{C_{q_2}}$ are in one-to-one mapping with the unitaries on the left- and right-hand side of \Cref{equ4:nqubit} respectively. Since, the unitaries are equivalent, the corresponding update functions are also equivalent.
\end{proof}

\subsection{Complexity of Computing the Update Function}
\label{sec:sizeofupdate}
Let $C_q$ be an $n$-qubit circuit  consisting  of a sequence of gates $g_1\ldots,g_{|C_q|}.$  The update function $F_{C_q}$ is computed by composing update functions for each gate in $C_q.$
									      $$ F_{C_q}= f_{g_{{|C_q|}}}\circ  \cdots \circ f_{g_2} \circ f_{g_1}.$$
Therefore the update function for any $n$-qubit quantum circuit $C_q$ with $k$  $\tgate$-gates can be defined as the following map.
\begin{equation}\label{nqubit:eq0}
\begin{aligned}
&F_{C_q}: \{0,1\}^{2n}\longrightarrow  (\mathbb{C} \times \{0,1\}^{2n})^{min(k,n)},\\
 &\vspace{0.25cm}(a_1,b_1,\ldots, a_n,b_n) \mapsto \left((\beta_1, {\bf s}_1),\dots, (\beta_{4^k}, {\bf s}_{4^k})\right).
 \end{aligned}
 \end{equation}
which corresponds to the following correction unitary
\begin{equation}
\label{exp:nqubit-correction}
\sum_{i=1}^{4^k} \beta_i \xgate^{b_{i_1}} \zgate^{a_{i_1}}\otimes \cdots \otimes \xgate^{b{i_n}} \zgate^{a_{i_n}}.
\end{equation}
where $\beta_i\in\mathbb{C}$ and ${\bf s}_i=a_{i_1}b_{i_1}, \ldots, a_{i_n}b_{i_n} \in\{0,1\}^{2n},$ $i\in[4^k].$ Therefore, in order to satisfy the efficiency requirement (polynomial-slowdown), we must have $k\in O(\log(|C_q|))$.
Note that the range of the update function can increase exponentially in the number of $\tgate$-gates as long as $k\leq n$ (\Cref{nqubit:eq0}). This is because there are at most $2^{2n}$ binary strings of length $2n,$ therefore for any $n$-qubit circuit, the correction unitary \Cref{exp:nqubit-correction} can be written as a summation of at most $2^{2n}$ terms. We will first prove that as long as the $\tgate$-count in $O(\log(|C_q|),$ the number of terms in $F_{C_q}$ has at most $O(|C_q|)$ terms.

\begin{lemma}\label{lem1:mainTh}
If $C_q$ is an $n$-qubit quantum circuit with $\tgate$-count in $O(\log(|C_q|),$ then the update function $F_{C_q}$ (\Cref{nqubit:eq0}) has at most $O(|C_q|)$ terms.
\end{lemma}
\begin{proof}
Let $C_q$ be an $n$-qubit circuit with $\tgate$-count in $O(\log(|C_q|).$ Suppose we want to evaluate $C_q$ on some $n$-qubit state $\ket{\psi},$ then after the step 4a of \Cref{QiO:qcircuit-teleportation}, we will obtain a state
\begin{equation}\label{nqubit:eq1}
C_q (\xgate^{a_{1}}\zgate^{b_{1}}\otimes \xgate^{a_{2}}\zgate^{b_{2}}\otimes \cdots \otimes \xgate^{a_{n}}\zgate^{b_{n}})\ket{\psi}.
\end{equation}
In order to recover $C_q(\ket{\psi})$ from the above expression, we multiply the correction unitary $U_{F_{C_q}}$ to the left hand side of expression \Cref{nqubit:eq1}. To compute $U_{F_{C_q}}$ we first compute the update function $F_{C_q}$ on the input $a_1b_1,\ldots,a_nb_n$

$$F_{C_q}(a_1b_1,\ldots,a_nb_n)=\left((\beta_1, {\bf s}_1),\dots, (\beta_{4^k}, {\bf s}_{4^k})\right)$$
where $\beta_i \in \mathbb{C},$ ${\bf s}_i\in\{0,1\}^{2n},$ $k\in \mathbb{N}$
$$U_{F_{C_q}}=\sum_{i=1}^{4^k} \beta_i \xgate^{b_{i_1}} \zgate^{a_{i_1}}\otimes \cdots \otimes \xgate^{b{i_n}} \zgate^{a_{i_n}}.$$

We will show that if the $\tgate$-count is in $O(\log(|C_q|),$ then $k\in O(|C_q|)$ (number of terms). Recall that
\begin{equation}\label{nqubit:eq2}
\tgate \xgate^b \zgate^a= (\alpha_1 \igate +  \alpha_2 \xgate + \alpha_3 \zgate + \alpha_4 \xgate\zgate)\tgate
\end{equation}

\begin{itemize}
\item {\bf Case 0}: Suppose $C_q$ has no $\tgate$-gates, then $C_q$ is a Clifford and there is only one term in $F_{C_q}.$ Therefore $k=0.$
\item {\bf Case 1}: Suppose $C_q$ has one $\tgate$-gate (acting on some $\ell$-th wire).
\begin{equation}\label{nqubit:eq3}
\begin{aligned}
&C_q (\xgate^{a_{1}}\zgate^{b_{1}}\otimes \cdots \otimes \xgate^{a_{n}}\zgate^{b_{n}})\\
&=(\xgate^{\otimes_{i=1}^{\ell-1} b_{i}^\prime} \zgate^{\otimes_{i=1}^{\ell-1}a_{i}^\prime})\otimes \left(\sum_{j=1}^{4} \beta_j \xgate^{b_j^\prime}\zgate^{a_j^\prime}\right) \otimes (\xgate^{\otimes_{i=\ell+1}^{n} b_{i}^\prime} \zgate^{\otimes_{i=\ell+1}^{n}a_{i}^\prime}) C_q\\
&=\sum_{i=1}^{4} \beta_i \xgate^{b_{i_1}^\prime} \zgate^{a_{i_1}^\prime}\otimes \cdots \otimes \xgate^{b_{i_n}^\prime} \zgate^{a_{i_n}^\prime}.
\end{aligned}
\end{equation}
Therefore $k\leq 4.$ It is important to realize that $4$ is the maximum number of terms a circuit with $\tgate$ can have. No Clifford gate including $\cnot$ can increase the number of terms beyond 4. This is because only a $\tgate$ gate can contribute a 4-term expression and then we expand all the terms in the correction unitary to maximum (\Cref{nqubit:eq3}). Now any Clifford acting on the correction unitary will act linearly on $U_{F_{C_q}}$ and only affect the bits $a_i^\prime$ and~$b_i^\prime$ (\Cref{nqubit:eq3}).
\item {\bf Case 2}: Similarly, if $C_q$ has two $\tgate$ gates, then each will contribute at most one expression of the form $(\beta_1 \igate +  \beta_2 \xgate + \beta_3 \zgate + \beta_4 \xgate\zgate).$  This will give us a correction unitary  $U_{F_{C_q}}=\sum_{i=1}^{4^2} \beta_i \xgate^{b_{i_1}} \zgate^{a_{i_1}}\otimes \cdots \otimes \xgate^{b{i_n}} \zgate^{a_{i_n}},$ therefore $k\leq 16.$  Note that it makes no difference whether $\tgate$ gates are acting on the same wire or on different wires for the worst-case analysis. If they are on the same wire, then the first $\tgate$ will contribute 4 terms and the second $\tgate$ will expand each term into 4 more terms, resulting in 16 terms. If they are acting on different wires, say $i$ and $j$ and suppose $\cnot$s are acting on the  $i$-th to $j$-th wire, then we may have to expand all terms in the unitary to apply $\cnot$s. This again can contribute at most 16 terms.
\item[] {\bf General Case}: Suppose  $C_q$  has $O(\log(|C_q|)$ $\tgate$ gates, then each~$\tgate$-gate will contribute at most a linear combination of 4 terms and in total at most $4^{O(\log(|C_q|)}$ terms. 
 therefore $k\in O(|C_q|).$\qedhere
\end{itemize}
\end{proof}

\begin{lemma}\label{lemma:nqubit:cost}
If $C_q$ be an $n$-qubit quantum circuit with $\tgate$-count $\in O(\log|C_q|),$ then there exists a classical circuit $C$ and a polynomial $p(\cdot)$ such that
\begin{itemize}
\item $C$ computes the update function $F_{C_q},$
\item $|C|\leq p(|C_q|).$
\end{itemize}
\end{lemma}
\begin{proof}
Let $C_q$ be a $n$-qubit quantum circuit with $\tgate$-count $\in O(\log|C_q|).$ Recall from \Cref{update function} that the corresponding update functions for the Clifford + $\tgate$ gate set are:
\begin{itemize}
\item $f_\xgate(a_1,b_1)=(a_1,b_1)$
\item $f_\zgate(a_1,b_1)=(a_1,b_1)$
\item $f_\hgate(a_1,b_1)=(b_1,a_1)$
\item $f_ \pgate(a_1,b_1)=(a_1,a_1\oplus b_1)$
\item $f_{\cnot}(a_1,b_1,a_2,b_2)= (a_1\oplus a_2,b_1,a_2, b_1\oplus b_2)$
\item $f_{\tgate}(a,b)=(\alpha_1,\alpha_2,\alpha_3,\alpha_4)$
\end{itemize}

Let $C_\xgate,$ $C_\zgate,$ $C_\hgate,$ $C_\pgate,$ $C_\cnot$ and $C_{\tgate}$ denote the classical circuits (called \emph{subcircuits})  that compute $f_\xgate, f_\zgate, f_\hgate f_ \pgate,  f_{\cnot}$ and $f_{\tgate}$ respectively. Clearly, these subcircuits are of constant size. Recall that all subcircuits for Cliffords map $k$-bit strings to $k$-bit strings ($k\in\{0,1\}$), but $C_{\tgate}$ expands its 2-bit input into a $4\ell$-bit strings (where $\ell =\max\{\left|\frac{1+ i}{2}\right|, \left|\frac{1-i}{2}\right|\}$\footnote{The notation $|a+bi|$  denotes the number of bits to represent the complex number $a+bi.$}. Let $C$ be a circuit that computes~$F_{C_q}$.  Then $C$ can be expressed in terms of these gadgets. To construct $C$, we go gate-by-gate in $C_q$ and employ the corresponding subcircuit. If $C_q$ is a Clifford circuit, then $C$ only consists of Clifford subcircuit (as mentioned earlier they map $k$ bits to $k$ bits $k\in\{2,4\}$  there are $O(|C_q|)$ gadgets in the circuit $C$). Otherwise if $C_q$ has $k$ $\tgate$-gates, then each $C_{\tgate}$ subcircuit will map a 2-bit string to a $4\ell$-bit string, potentially increasing the size of $C$  to $O({(4\ell)}^k),$ but since  $k\in O(\log(|C_q|),$ we have $O({(4\ell)}^k)\in O(|C_q|).$ Therefore, there exists a polynomial $p(\cdot)$ such that $|C|\leq p(n).$
\end{proof}

\section{Quantum Indistinguishability Obfuscation with Respect to a Pseudo-Distance}
\label{sec:quantum:iO:approx:circuits}

\label{sec:approx:circuits}
In this section, we provide a definition for circuits that are approximately equivalent  (with respect to a pseudo-distance) (\Cref{def:aqec}). In \Cref{sec:qiO:approx:circuits}, we present a definition of quantum indistinguishability obfuscation with respect to a pseudo-distance, and in  \Cref{sec:Gottesman-Chuang}, we present a scheme that satisfies this definition,   for circuits close to a fixed level of the Gottesman-Chuang hierarchy.

\subsection{Approximately Equivalent Quantum Circuits}
\begin{definition}{\rm (Approximately Equivalent Quantum Circuits):}
\label{def:aqec}
Let $C_{q_0}$ and~$C_{q_1}$ be two $n$-qubit quantum circuits and {\bf D} be a pseudo-distance. We say $C_{q_0}$ and $C_{q_1}$ are \emph{approximately equivalent} with respect to {\bf D} if there exists a negligible function ${\tt negl}(n)$ such that  $${\bf D}(C_{q_0}, C_{q_1})\leq {\tt negl}(n).$$
\end{definition}

\subsection{Indistinguishability Obfuscation for Approximately Equivalent Quantum Circuits}
\label{sec:qiO:approx:circuits}
In this section, we provide a definition of quantum indistinguishability obfuscation for approximately equivalent circuits, $qi\mathcal{O}_{\bf D}.$ To be consistent with \Cref{def:QiO}, we  require that the obfuscator, on input a quantum circuit $C_q,$  outputs an auxiliary quantum state $\ket{\phi}$  and a quantum circuit $C_q^\prime$, but note in the actual construction (\Cref{QiO:gottesman-chuang}), the state $\ket{\phi}$ is an empty register. Here, we consider only the case of \emph{statistical} security. Notable here is the indistinguishability property is required to hold not only for equivalent quantum circuits, but also for \emph{approximately} equivalent quantum circuits. Also, contrary to \Cref{def:QiO}, we only require the indistinguishability for large values of~$n$.

\begin{definition}
\label{def:aQiO}
Let $\mathcal{C}_Q$ be a polynomial-time family of reversible quantum circuits and let {\bf D} be a pseudo-distance. For $n\in\mathbb{N}$, let $C_{q^n}$ be the circuits in $\mathcal{C}_Q$ of input length $n.$
A  polynomial-time quantum algorithm for~$\mathcal{C}_Q$ is a \emph{statistically secure quantum indistinguishability obfuscator} ($qi\mathcal{O}_{\bf D}$) for $\mathcal{C}_Q$  \emph{with respect to {\bf D}} if the following conditions hold:

\begin{enumerate}
\item {\tt Functionality:} There exists a negligible function ${\tt negl}(n)$ such that for every $C_q\in C_{q^n}$
$$(\ket{\phi}, C_q^\prime)\leftarrow qi\mathcal{O}_{\bf D}(C_q)  \;  \mbox{ and }\;   \mathbf{ D}(C_q^\prime(\ket{\phi}, \cdot),C_{q}(\cdot))\leq {\tt negl}(n).$$
Where $\ket{\phi}$ is an $\ell$-qubit state, the circuits $C_q$ and $C_q^\prime$ are of type $(n,n)$ and $(m,n)$ respectively ($m= \ell +n$).\footnote{A circuit is of type $(i,j)$ if it maps $i$ qubits to $j$ qubits.}

\item  {\tt Polynomial Slowdown:} There exists a polynomial $p(n)$ such that for any $C_{q}\in C_{q^n},$
\begin{itemize}
\item  $\ell\leq p(|C_{q}|)$
\item $m \leq  p(|C_{q}|)$
\item $|C_{q}^\prime| \leq p(|C_{q}|).$
\end{itemize}

\item {\tt Statistically Secure Indistinguishability:} For any two \textbf{approximately equivalent} quantum circuits $C_{q_0},C_{q_1}\in C_{q^n},$ of the same size \textbf{and for large enough $n,$} the two distributions $qi\mathcal{O}_{\bf D}(C_{q_0})$ and $qi\mathcal{O}_{\bf D}(C_{q_1})$ are statistically indistinguishable.			
\end{enumerate}
\end{definition}

\subsection{$qi\mathcal{O}_{{\bf D}}$ for Circuits Close to the Gottesman-Chuang Hierarchy }
\label{sec:Gottesman-Chuang}

Here, we present a quantum indistinguishability obfuscation (\cref{def:QiO}) for a family of circuits that are approximately equivalent (\Cref{def:aqec}) with respect to the pseudo-distance ${\bf D}(U_1,U_2)=\frac{1}{\sqrt{2d^2}} ||U_1\otimes U_1^*- U_2\otimes U_2^*||_F \;$ (see \Cref{sec:norms}). There are two main ingredient in our construction, one is  Low's learning algorithm~\cite{Low09} (described below)
 and the second is~\Cref{lem:approx}.

In \cite{Low09} Low presents a learning algorithm that, given  oracle access to a unitary~$U$ and its conjugate $U^\dagger$ with the promise that the distance ${\bf D}(U, C)\leq \epsilon <\frac{1}{2^{k-1/2}}$  for some $C\in \mathcal{C}_k$ (\Cref{sec:gottesman-chuang}), outputs a circuit $C_q$ for computing $C$ with probability at least $1-\delta$ with $$O\left(\frac{1}{{\epsilon^\prime}^2} (2n)^{k-1} \log\left(\frac{(2n+1)^{k-1})}{\delta}\right)\right)$$ queries. Where $\epsilon^\prime:=\sqrt{2(1-(2^{k-1}\epsilon)^2}-1>0$ and ${\bf D}(U_1,U_2)=\frac{1}{\sqrt{2d^2}} ||U_1\otimes U_1^*- U_2\otimes U_2^*||_F$ is the pseudo-distance defined in \Cref{sec:norms}.

Based on Low's work, we construct an quantum indistinguishability obfuscation $qi\mathcal{O}_{{\bf D}}$ with respect to this pseudo-distance ${\bf D}$ for circuits that are very close to $\mathcal{C}_k.$ Note that the run-time of Low's algorithm is exponential in $k.$ Moreover, the algorithm becomes infeasible if $\epsilon^\prime$ is very small. Therefore, to ensure that our construction in \Cref{QiO:gottesman-chuang} runs in polynomial-time we set $k$ to be some fixed positive integer and $\epsilon \leq {\tt negl}(n)<\frac{1}{2^{k-1/2}}$ for all $n.$ Note if $\epsilon<\frac{1}{2^{k-1/2}},$ then $ \epsilon^\prime \geq \frac{\sqrt{7}}{2}-1.$

\begin{lemma}
\label{lem:approx}
Let $U$ and $C$ be unitaries. If the distance ${\bf D}(U,C)<\frac{1}{2^{k-1/2}}$ for some $C\in \mathcal{C}_k$, then $C$ is unique up to phase.
\end{lemma}
\begin{proof}
See~\cite{Low09}.
\end{proof}

\begin{theorem}\label{th:approx}
Let $\mathcal{C}_Q=\{U_{q^{n,k}} \mid n \in \mathbb{N} \mbox{ and } k \mbox{ is fixed positive integer} \},$ be a polynomial-time family of reversible quantum circuits. Here, $U_{q^{n,k}}$ denotes the $n$-qubit circuits for which there exists a negligible function ${\texttt{negl}}(n)$ such that for any $U_q\in U_{q^{n,k}},$ there exists a $C_q\in\mathcal{C}_k$  that satisfies ${\bf D}(U_q, C_q)< {\tt negl}(n)<\frac{1}{2^{k+1/2}}.$ Then \Cref{QiO:gottesman-chuang} is a statistically-secure quantum indistinguishability obfuscation for $\mathcal{C}_Q$ with respect to~${\bf D}.$
\end{theorem}

\begin{algorithm}[]
   \caption{$qi\mathcal{O}$-Gottesman-Chuang}
   \label{QiO:gottesman-chuang}
  \begin{itemize}
\item Input: An $n$-qubit circuit $U_q\in U_{q^{n,k}},$  $k$  and $\delta={\tt negl}(n)).$
 \begin{enumerate}
 \item  From $U_q$ compute the circuit $U_q^\dagger.$
  \item Using Low's approximate learning algorithm on inputs $U_q$ and ${U_q}^\dagger$ compute the circuit $C_q $~\cite{Low09}.
  \item Output the circuit $C_q.$
\end{enumerate}
  \end{itemize}
\end{algorithm}

\begin{proof} We have to show that \Cref{QiO:gottesman-chuang} satisfies \Cref{def:aQiO}.
\begin{enumerate}

\item  {\tt Functionality:}  On input $U_q\in U_{q^{n,k}}$, let $C_q$ be the output of \Cref{QiO:gottesman-chuang}. By assumption (\Cref{th:approx}) there exists a unitary $C \in \mathcal{C}_k$ such that ${\bf D}(U_q, C)< {\tt negl}(n)<\frac{1}{2^{k+1/2}}.$ From \Cref{lem:approx}, $C$ is unique up to a global phase. Therefore with overwhelming probability, Low's approximate learning algorithm will output a circuit $C_q$ that computes~$C$. We have ${\bf D}(U_q, C_q)={\bf D}(U_q, C_q)< {\tt negl}(n),$ therefore $U_q$ and $C_q$ are approximately equivalent.

\item {\tt Polynomial Slowdown:} The total  cost of \Cref{QiO:gottesman-chuang} is the cost of computing the circuit $U_q^\dagger$ from $U_q$ plus the cost of Low's learning algorithm. Clearly, we can compute $U_q^\dagger$ from $U_q^\dagger$ in polynomial-time. Using Low's algorithm (for parameters defined in \Cref{th:approx} and setting $\delta={\tt negl}(n)$) we can learn $C_q$ with probability at least $1-{\tt negl}(n)$ in time at most $$O\left(n^{k-1} \left[(k-1)\log((2n+1))-\log({\tt negl}(n))\right]\right).$$ Which is at most a polynomial in $n$  (since $k$ is a constant). Therefore, \Cref{QiO:gottesman-chuang} runs in polynomial-time.

\item {\tt Statistically Indistinguishability:} Let $U_q, U_q^\prime\in U_{q^{n,k}}$ be two circuits such that $${\bf D}(U_q, U_q^\prime)<{\tt negl}(n).$$
By assumption (\Cref{th:approx}), there exist unitaries $C,C^\prime\in \mathcal{C}_k$ such that ${\bf D}(U_q, C)<{\tt negl}(n)<\frac{1}{2^{k+1/2}}$ and  ${\bf D}(U_q^\prime, C^\prime)<{\tt negl}(n)<\frac{1}{2^{k+1/2}}.$ Using the triangle inequality we can easily show that $C$ and $C^\prime$ are equivalent circuit (up to a global phase) for large $n.$

$${\bf D}(U_q, C^\prime)\leq {\bf D}(U_q, U_q^\prime)+{\bf D}(U_q^\prime, C^\prime)\leq {\tt negl}(n)+\frac{1}{2^{k+1/2}}$$
$$\Longrightarrow {\bf D}(U_q, C^\prime)<\frac{1}{2^{k+1/2}}+\frac{1}{2^{k+1/2}}<\frac{1}{2^{k-1/2}}$$

According \Cref{lem:approx}  $C^\prime$ is unique (up to a global phase) such that ${\bf D}(U_q, C^\prime)<\frac{1}{2^{k-1/2}},$ But $C$ also satisfies ${\bf D}(U_q, C)<\frac{1}{2^{k-1/2}}.$ It follows that $C$ and $C^\prime$ are equivalent circuits (up to a global phase). Moreover, the outputs of Low's algorithm $C_q$ and $C_q^\prime$ on any two equivalent unitaries (up to global phase) is statistically indistinguishable~\cite{Low09}. Therefore, for any  $U_q, U_q^\prime\in U_{q^{n,k}},$ the two distributions $qi\mathcal{O}_{\bf D}(U_q)$ and $qi\mathcal{O}_{\bf D}(U_q^\prime)$ are statistically indistinguishable.	
\qedhere
\end{enumerate}
\end{proof}

\section*{Acknowledgements}
We thank an anonymous reviewer for pointing out the work of~\cite{Low09}; we would also like to thank Yfke Dulek for related discussions.
This material is based upon work supported by the Air Force Office of Scientific Research under award number
 FA9550-20-1-0375,
 Canada's   NFRF and NSERC, an Ontario ERA, and the University of Ottawa’s Research Chairs program.

\appendix
\section{Size of Coefficients in the Update Functions}
\label{sec:appendix}
 \label{coeff:size}
Here, we prove a Lemma that is used in \Cref{sec:1-qubit}.
\begin{lemma}\label{lem:beta:size}
Let $C_q$ be an  $n$-qubit  $poly(n)$ size quantum circuit with $O(\log(n))$ number of gates and  $F_{C_q}$ be the corresponding update function
\begin{equation*}
\begin{aligned}
&F_{C_q}: \{0,1\}^{2n}\longrightarrow  (\mathbb{C} \times \{0,1\}^{2n})^{min(k,n)},\\
 &\vspace{0.25cm}(a_1,b_1,\ldots, a_n,b_n) \mapsto \left((\beta_1, {\bf s}_1),\dots, (\beta_{4^k}, {\bf s}_{4^k})\right).
 \end{aligned}
 \end{equation*}
  then there exists a polynomial $p(\cdot)$ such $\beta_i \in O(p(n))$ for all $i\in[4^k].$
\end{lemma}

\begin{proof}
The following map is an isomorphism between $\mathbb{C}$ and $\mathbb{R}^2$
 \begin{equation}
  \label{size:map:real-complex)}
  f:\mathbb{C}\leftarrow \mathbb{R}^2, \; (a+b_i)\mapsto (a,b)
\end{equation}

Note that there is a one-to-one map between $F_{C_q}$ and the corresponding unitary $U_{F_{C_q}}=\sum_{i=1}^{4^k} \beta_i \xgate^{b_{i_1}} \zgate^{a_{i_1}}\otimes \cdots \otimes \xgate^{b{i_n}} \zgate^{a_{i_n}},\;  \; k\leq n.$ Note that any Clifford gate can only affect the correction bits in unitaries of type $U_{F_{C_q}},$ but will have no effect on the coefficients $\beta_i.$ So the coefficients can be affected by $\tgate$ gates. Therefore, to estimate the size of the coefficients we can ignore other gates. Each $\beta_i$ is constructed by adding and multiplying numbers from the set $\left\{0,1, \frac{1+i}{2},\frac{1-i}{2}\right\}.$ We note the following relationships between $\frac{1+i}{2},\frac{1-i}{2}$
\begin{equation*}
 \left(\frac{1+i}{2}\right)\pm \left(\frac{1-i}{2}\right)=\pm1.
\end{equation*}
If $ m=2\ell+1$ and  $\ell\in\mathbb{N},$ then
 \begin{equation*}
 \hspace{4cm}\left(\frac{1\pm i}{2}\right)^m=\left(\frac{a}{2}\right)^\ell,\;  a\in\{\pm 1,\pm i\}.
\end{equation*}
Else if  $m=2\ell$ and $\ell\in\mathbb{N}\cup \{0\},$ then
\begin{equation*}
 \left(\frac{1\pm i}{2}\right)^m=\left(\frac{a}{2}\right)^\ell \left(\frac{1\pm i}{2}\right),\; a\in\{\pm 1,\pm i\}
\end{equation*}

Therefore, we can represent $\left(\frac{1\pm i}{2}\right)^m$ in $O(m)$ bits  and $\left(\frac{1+ i}{2}\right)^{m} \left(\frac{1- i}{2}\right)^{m}$ in $O(m)$ bits. Of course $1^m=1$ and adding $1$ to itself is $m.$ For each application of $\tgate$, a coefficient will multiply and add at most polynomial time in the circuit size and there are $O(\log(n))$ such gates, therefore there exists a polynomial $p(\cdot)$ such that $|\beta_i|\in O(p(n)),$ for every $i\in[4^k].$ 
\end{proof}

\addcontentsline{toc}{section}{References}

\bibliographystyle{bst/arxiv_no_month}
\bibliography{bib/full,bib/quantum,bib/more}

\begin{thebibliography}{10}

\bibitem{Aar09}
S.~Aaronson.
\newblock Quantum copy-protection and quantum money.
\newblock In {\em 24th Annual Conference on Computational Complexity---CCC
  2009}, pages 229--242, 2009.
\newblock
  \texttt{\href{http://dx.doi.org/10.1109/CCC.2009.42}{DOI:\,10.1109/CCC.2009.42}}.

\bibitem{AG04}
S.~Aaronson and D.~Gottesman.
\newblock Improved simulation of stabilizer circuits.
\newblock {\em Physical Review A}, 70(5):052328, 2004.
\newblock
  \texttt{\href{http://dx.doi.org/10.1103/PhysRevA.70.052328}{DOI:\,10.1103/PhysRevA.70.052328}}.

\bibitem{AF16arxiv}
G.~Alagic and B.~Fefferman.
\newblock On quantum obfuscation, 2016.
\newblock Available at \url{https://arxiv.org/abs/1602.01771}.

\bibitem{AJJ14}
G.~Alagic, S.~Jeffery, and S.~Jordan.
\newblock Circuit obfucation using braids.
\newblock In {\em 9th Conference on the Theory of Quantum Computation,
  Communication and Cryptography---TQC 2014}, pages 141--160, 2014.
\newblock
  \texttt{\href{http://dx.doi.org/10.4230/LIPIcs.TQC.2014.141}{DOI:\,10.4230/LIPIcs.TQC.2014.141}}.

\bibitem{ABD16}
M.~Albrecht, S.~Bai, and L.~Ducas.
\newblock A subfield lattice attack on overstretched {NTRU} assumptions.
\newblock In {\em Advances in Cryptology---CRYPTO 2016}, volume~1, pages
  153--178, 2016.
\newblock
  \texttt{\href{http://dx.doi.org/10.1007/978-3-662-53018-4\_6}{DOI:\,10.1007/978-3-662-53018-4\_6}}.

\bibitem{AMTW00}
A.~Ambainis, M.~Mosca, A.~Tapp, and R.~de~Wolf.
\newblock Private quantum channels.
\newblock In {\em 41st Annual Symposium on Foundations of Computer
  Science---FOCS 2000}, pages 547--553, 2000.
\newblock
  \texttt{\href{http://dx.doi.org/10.1109/SFCS.2000.892142}{DOI:\,10.1109/SFCS.2000.892142}}.

\bibitem{AMM14}
M.~Amy, D.~Maslov, and M.~Mosca.
\newblock Polynomial-time {$T$}-depth optimization of {C}lifford+{$T$} circuits
  via matroid partitioning.
\newblock {\em IEEE Transactions on Computer-Aided Design of Integrated
  Circuits and Systems}, 33(10):1476--1489, 2014.
\newblock
  \texttt{\href{http://dx.doi.org/10.1109/TCAD.2014.2341953}{DOI:\,10.1109/TCAD.2014.2341953}}.

\bibitem{AMMR13}
M.~Amy, D.~Maslov, M.~Mosca, and M.~Roetteler.
\newblock A meet-in-the-middle algorithm for fast synthesis of depth-optimal
  quantum circuits.
\newblock {\em IEEE Transactions on Computer-Aided Design of Integrated
  Circuits and Systems}, 32(6):818--830, 2013.
\newblock
  \texttt{\href{http://dx.doi.org/10.1109/TCAD.2013.2244643}{DOI:\,10.1109/TCAD.2013.2244643}}.

\bibitem{AJL+19}
P.~Ananth, A.~Jain, H.~Lin, C.~Matt, and A.~Sahai.
\newblock Indistinguishability obfuscation without multilinear maps: New
  paradigms via low degree weak pseudorandomness and security amplification.
\newblock In {\em Advances in Cryptology---CRYPTO 2019}, volume~3, pages
  284--332, 2019.
\newblock
  \texttt{\href{http://dx.doi.org/10.1007/978-3-030-26954-8\_10}{DOI:\,10.1007/978-3-030-26954-8\_10}}.

\bibitem{AP20}
P.~Ananth and R.~L.~L. Placa.
\newblock Secure software leasing, 2020.
\newblock https://arxiv.org/abs/2005.05289.

\bibitem{BGI+12}
B.~Barak, O.~Goldreich, R.~Impagliazzo, S.~Rudich, A.~Sahai, S.~Vadhan, and
  K.~Yang.
\newblock On the (im)possibility of obfuscating programs.
\newblock {\em Journal of the ACM}, 59(2):6, 2012.
\newblock
  \texttt{\href{http://dx.doi.org/10.1145/2160158.2160159}{DOI:\,10.1145/2160158.2160159}}.

\bibitem{BBC+93}
C.~H. Bennett, G.~Brassard, C.~Cr{\'e}peau, R.~Jozsa, A.~Peres, and W.~K.
  Wootters.
\newblock Teleporting an unknown quantum state via dual classical and
  {Einstein-Podolsky-Rosen} channels.
\newblock {\em Physical Review Letters}, 70(13):1895, 1993.
\newblock
  \texttt{\href{http://dx.doi.org/10.1103/PhysRevLett.70.1895}{DOI:\,10.1103/PhysRevLett.70.1895}}.

\bibitem{BP15}
N.~Bitansky and O.~Paneth.
\newblock {ZAP}s and non-interactive witness indistinguishability from
  indistinguishability obfuscation.
\newblock In {\em 12th Theory of Cryptography Conference---TCC 2015},
  volume~II, pages 401--427, 2015.
\newblock
  \texttt{\href{http://dx.doi.org/10.1007/978-3-662-46497-7\_16}{DOI:\,10.1007/978-3-662-46497-7\_16}}.

\bibitem{BZ14}
D.~Boneh and M.~Zhandry.
\newblock Multiparty key exchange, efficient traitor tracing, and more from
  indistinguishability obfuscation.
\newblock In {\em Advances in Cryptology---CRYPTO 2014}, volume~I, pages
  480--499, 2014.
\newblock
  \texttt{\href{http://dx.doi.org/10.1007/978-3-662-44371-2\_27}{DOI:\,10.1007/978-3-662-44371-2\_27}}.

\bibitem{Bra18}
Z.~Brakerski.
\newblock Quantum {FHE} (almost) as secure as classical.
\newblock In {\em Advances in Cryptology---CRYPTO 2018}, volume~3, pages
  67--95, 2018.
\newblock
  \texttt{\href{http://dx.doi.org/10.1007/978-3-319-96878-0\_3}{DOI:\,10.1007/978-3-319-96878-0\_3}}.

\bibitem{BJ15}
A.~Broadbent and S.~Jeffery.
\newblock Quantum homomorphic encryption for circuits of low {T}-gate
  complexity.
\newblock In {\em Advances in Cryptology---CRYPTO 2015}, volume~2, pages
  609--629, 2015.
\newblock
  \texttt{\href{http://dx.doi.org/10.1007/978-3-662-48000-7\_30}{DOI:\,10.1007/978-3-662-48000-7\_30}}.

\bibitem{BL19arxiv}
A.~Broadbent and S.~Lord.
\newblock Uncloneable quantum encryption via oracles.
\newblock In {\em Theory of Quantum Computation, Communication, and
  Cryptography---TQC 2020}, pages 4:1--4:22, 2020.
\newblock
  \texttt{\href{http://dx.doi.org/10.4230/LIPIcs.TQC.2020.4}{DOI:\,10.4230/LIPIcs.TQC.2020.4}}.

\bibitem{CLTV15}
R.~Canetti, H.~Lin, S.~Tessaro, and V.~Vaikuntanathan.
\newblock Obfuscation of probabilistic circuits and applications.
\newblock In {\em 12th Theory of Cryptography Conference---TCC 2015},
  volume~II, pages 468--497, 2015.
\newblock
  \texttt{\href{http://dx.doi.org/10.1007/978-3-662-46497-7\_19}{DOI:\,10.1007/978-3-662-46497-7\_19}}.

\bibitem{CGH17}
Y.~Chen, C.~Gentry, and S.~Halevi.
\newblock Cryptanalyses of candidate branching program obfuscators.
\newblock In {\em Advances in Cryptology---EUROCRYPT 2017}, volume~3, pages
  278--307, 2017.
\newblock
  \texttt{\href{http://dx.doi.org/10.1007/978-3-319-56617-7\_10}{DOI:\,10.1007/978-3-319-56617-7\_10}}.

\bibitem{CLT13}
J.-S. Coron, T.~Lepoint, and M.~Tibouchi.
\newblock Practical multilinear maps over the integers.
\newblock In {\em Advances in Cryptology---CRYPTO 2013}, volume~1, pages
  476--493, 2013.
\newblock
  \texttt{\href{http://dx.doi.org/10.1007/978-3-642-40041-4\_26}{DOI:\,10.1007/978-3-642-40041-4\_26}}.

\bibitem{CDPR16}
R.~Cramer, L.~Ducas, C.~Peikert, and O.~Regev.
\newblock Recovering short generators of principal ideals in cyclotomic rings.
\newblock In {\em Advances in Cryptology---EUROCRYPT 2016}, volume~2, pages
  559--585, 2016.
\newblock
  \texttt{\href{http://dx.doi.org/10.1007/978-3-662-49896-5\_20}{DOI:\,10.1007/978-3-662-49896-5\_20}}.

\bibitem{DMM16}
O.~Di~Matteo and M.~Mosca.
\newblock Parallelizing quantum circuit synthesis.
\newblock {\em Quantum Science and Technology}, 1(1):015003, 2016.
\newblock
  \texttt{\href{http://dx.doi.org/10.1088/2058-9565/1/1/015003}{DOI:\,10.1088/2058-9565/1/1/015003}}.

\bibitem{DSS16}
Y.~Dulek, C.~Schaffner, and F.~Speelman.
\newblock Quantum homomorphic encryption for polynomial-sized circuits.
\newblock In {\em Advances in Cryptology---CRYPTO 2016}, pages 3--32, 2016.
\newblock
  \texttt{\href{http://dx.doi.org/10.1007/978-3-662-53015-3\_1}{DOI:\,10.1007/978-3-662-53015-3\_1}}.

\bibitem{GGH+13}
S.~Garg, C.~Gentry, S.~Halevi, M.~Raykova, A.~Sahai, and B.~Waters.
\newblock Candidate indistinguishability obfuscation and functional encryption
  for all circuits.
\newblock In {\em 54th Annual Symposium on Foundations of Computer
  Science---FOCS 2013}, pages 40--49, 2013.
\newblock
  \texttt{\href{http://dx.doi.org/10.1109/FOCS.2013.13}{DOI:\,10.1109/FOCS.2013.13}}.

\bibitem{GGH15}
C.~Gentry, S.~Gorbunov, and S.~Halevi.
\newblock Graph-induced multilinear maps from lattices.
\newblock In {\em 12th Theory of Cryptography Conference---TCC 2015}, volume~2,
  pages 498--527, 2015.
\newblock
  \texttt{\href{http://dx.doi.org/10.1007/978-3-662-46497-7\_20}{DOI:\,10.1007/978-3-662-46497-7\_20}}.

\bibitem{GS19arxiv}
B.~Giles and P.~Selinger.
\newblock Remarks on {M}atsumoto and {A}mano’s normal form for single-qubit
  {C}lifford+${T}$ operators, 2019.
\newblock Available at \url{https://arxiv.org/abs/1312.6584}.

\bibitem{GR14}
S.~Goldwasser and G.~N. Rothblum.
\newblock On best-possible obfuscation.
\newblock {\em Journal of Cryptology}, 27(3):480--505, 2014.
\newblock
  \texttt{\href{http://dx.doi.org/10.1007/s00145-013-9151-z}{DOI:\,10.1007/s00145-013-9151-z}}.

\bibitem{ABDS20}
Y.~D. C.~S. Gorjan~Alagic, Zvika~Brakerski.
\newblock Impossibility of quantum virtual black-box obfuscation of classical
  circuits, 2020.
\newblock https://arxiv.org/abs/2005.06432.

\bibitem{Got98}
D.~Gottesman.
\newblock The {H}eisenberg representation of quantum computers.
\newblock In {\em 22nd International Colloquium on Group Theoretical Methods in
  Physics---GROUP 22}, pages 32--43, 1998.
\newblock
  \texttt{\href{http://arxiv.org/abs/quant-ph/9807006}{arXiv:\,quant-ph/9807006}}.

\bibitem{GC99}
D.~Gottesman and I.~L. Chuang.
\newblock Demonstrating the viability of universal quantum computation using
  teleportation and single-qubit operations.
\newblock {\em Nature}, 402:390--393, 1999.
\newblock \texttt{\href{http://dx.doi.org/10.1038/46503}{DOI:\,10.1038/46503}}.

\bibitem{GMOR15}
S.~Guo, T.~Malkin, I.~C. Oliveira, and A.~Rosen.
\newblock The power of negations in cryptography.
\newblock In {\em 12th Theory of Cryptography Conference---TCC 2015}, volume~1,
  pages 36--65, 2015.
\newblock
  \texttt{\href{http://dx.doi.org/10.1007/978-3-662-46494-6\_3}{DOI:\,10.1007/978-3-662-46494-6\_3}}.

\bibitem{JLS20}
A.~Jain, H.~Lin, and A.~Sahai.
\newblock Indistinguishability obfuscation from well-founded assumptions, 2020.
\newblock Available at \url{https://eprint.iacr.org/2020/1003}.

\bibitem{KLM07}
P.~Kaye, R.~Laflamme, and M.~Mosca.
\newblock {\em An introduction to quantum computing}.
\newblock Oxford, 2007.

\bibitem{LSS14}
A.~Langlois, D.~Stehl{\'e}, and R.~Steinfeld.
\newblock {GGHLite}: More efficient multilinear maps from ideal lattices.
\newblock In {\em Advances in Cryptology---EUROCRYPT 2014}, pages 239--256,
  2014.
\newblock
  \texttt{\href{http://dx.doi.org/10.1007/978-3-642-55220-5\_14}{DOI:\,10.1007/978-3-642-55220-5\_14}}.

\bibitem{Low09}
R.~A. Low.
\newblock Learning and testing algorithms for the {C}lifford group.
\newblock {\em Physical Review A}, 80(5):052314, 2009.
\newblock
  \texttt{\href{http://dx.doi.org/https://doi.org/10.1103/PhysRevA.80.052314}{DOI:\,https://doi.org/10.1103/PhysRevA.80.052314}}.

\bibitem{MA08arxiv}
K.~Matsumoto and K.~Amano.
\newblock Representation of quantum circuits with {C}lifford and $\pi/8$ gates,
  2008.
\newblock Available at \url{https://arxiv.org/abs/0806.3834}.

\bibitem{NWD14}
P.~Niemann, R.~Wille, and R.~Drechsler.
\newblock Efficient synthesis of quantum circuits implementing {C}lifford group
  operations.
\newblock In {\em 19th Asia and South Pacific Design Automation
  Conference---ASP-DAC 2014}, pages 483--488, 2014.
\newblock
  \texttt{\href{http://dx.doi.org/10.1109/ASPDAC.2014.6742938}{DOI:\,10.1109/ASPDAC.2014.6742938}}.

\bibitem{GJLS21}
H.~L. Romain~Gay, Aayush~Jain and A.~Sahai.
\newblock Indistinguishability obfuscation from simple-to-state hardness
  assumptions.
\newblock 2021.

\bibitem{SW14}
A.~Sahai and B.~Waters.
\newblock How to use indistinguishability obfuscation: deniable encryption, and
  more.
\newblock In {\em 46th Annual ACM Symposium on Theory of Computing---STOC
  2014}, pages 475--484, 2014.
\newblock
  \texttt{\href{http://dx.doi.org/10.1145/2591796.2591825}{DOI:\,10.1145/2591796.2591825}}.

\bibitem{Sel13arxiv}
P.~Selinger.
\newblock Generators and relations for {$n$}-qubit {C}lifford operators, 2013.
\newblock Available at \url{https://arxiv.org/abs/1310.6813}.

\bibitem{Sip12}
M.~Sipser.
\newblock {\em Introduction to the Theory of Computation}.
\newblock Cengage Learning, 3rd edition, 2012.

\bibitem{Spe16}
F.~Speelman.
\newblock Instantaneous non-local computation of low {$T$}-depth quantum
  circuits.
\newblock In {\em 11th Conference on the Theory of Quantum Computation,
  Communication and Cryptography---TQC 2016}, pages 9:1--9:24, 2016.
\newblock
  \texttt{\href{http://dx.doi.org/10.4230/LIPIcs.TQC.2016.9}{DOI:\,10.4230/LIPIcs.TQC.2016.9}}.

\end{thebibliography}
\end{document}